%% file: Consensus-solvability.tex
\documentclass[manuscript,nonacm]{acmart}

\usepackage[T1]{fontenc}
\usepackage{amsmath}
\usepackage[]{todonotes}
\usepackage[ruled,linesnumbered,noend]{algorithm2e}
\usepackage{wrapfig}

\usepackage{comment}
\usepackage{tikzit}
\usepackage[capitalize]{cleveref}
\usepackage[appendix=append]{apxproof} 

\makeatletter
\newtheorem*{rep@theorem}{\rep@title}
\newcommand{\newreptheorem}[2]{%
	\newenvironment{rep#1}[1]{%
		\def\rep@title{#2~\ref{##1}}%
		\begin{rep@theorem}}%
		{\end{rep@theorem}}}
\makeatother

\newenvironment{lemma-repeat}[1]{\begin{trivlist}
		\item[\hspace{\labelsep}{\bf\noindent Lemma~\ref{#1} }]\em }%
	{\end{trivlist}}
\newenvironment{theorem-repeat}[1]{\begin{trivlist}
		\item[\hspace{\labelsep}{\bf\noindent Theorem~\ref{#1} }]\em }%
	{\end{trivlist}}

\newcommand{\qedsymb}{\qed}

\newtheoremrep{theorem}{Theorem}
\newtheoremrep{claim}{Claim}
\newtheoremrep{lemma}{Lemma}
\newtheoremrep{corollary}{Corollary}

\newtheorem{definition}{Definition}

\input{drawings/Worst-Case-Graphs.tikzstyles}

\DeclareMathOperator{\view}{view}

\DeclareMathOperator{\In}{In}

\DeclareMathOperator{\Root}{Root}

\DeclareMathOperator{\TD}{TD}

\newcommand{\N}{\mathcal{N}}

\newcommand{\bD}{\mathbf{D}}

\newcommand{\bC}{\mathbf{C}}
\newcommand{\bG}{\mathbf{G}}

\newcommand{\bS}{\mathbf{S}}

\newcommand{\cC}{\mathcal{C}}

\newcommand{\cS}{\mathcal{S}}

\title[Time Complexity of Consensus in Dynamic Networks]{Time Complexity of Consensus in Dynamic Networks\\Under Oblivious Message Adversaries}



\author{Ami Paz}
\affiliation{%
	\institution{LISN --- CNRS \& Université Paris-Saclay}
	\country{France}}
\email{ami.paz@lisn.fr}

\author{Hugo Rincon Galeana}
\orcid{0000-0002-8152-1275}
\affiliation{%
	\institution{Technische Universität Wien}
	\department{ECS Group}
	\city{Vienna}
	\country{Austria}}
\email{hugorincongaleana@gmail.com}

\author{Stefan Schmid}
\orcid{0000-0002-7798-1711}
\affiliation{%
	\institution{Technische Universität Berlin, Germany \& Universität Wien}
	\department{Faculty IV}
	\city{Vienna}
	\country{Austria}}
\email{stefan.schmid@tu-berlin.de}

\author{Ulrich Schmid}
\orcid{0000-0001-9831-8583}
\affiliation{%
	\institution{Technische Universität Wien}
	\department{ECS Group}
	\city{Vienna}
	\country{Austria}}
\email{s@ecs.tuwien.ac.at}

\author{Kyrill Winkler}
\orcid{0000-0002-7310-1748}
\affiliation{%
	\institution{Universität Wien}
	\department{Faculty of Computer Science}
	\city{Vienna}
	\country{Austria}}
\email{kyrill.winkler@gmail.com}

\begin{document}
  
\begin{abstract}
Consensus is a most fundamental task in distributed computing. This paper studies the consensus problem for a set of processes connected by a dynamic directed network, in which computation and communication is lock-step synchronous but controlled by an oblivious message adversary.  In this basic model, determining consensus solvability and designing consensus algorithms in the case where it is possible,  has been shown to be surprisingly difficult. We present an explicit decision procedure to determine if consensus is possible under a given adversary. This in turn enables us, for the first time, to study the time complexity of consensus in this model. In particular, we derive time complexity upper bounds for consensus solvability both for a centralized decision procedure as well as for solving distributed consensus. We complement these results with time complexity lower bounds. Intriguingly, we find that reaching consensus under an oblivious message adversary can take exponentially longer than broadcasting the input value of some process to all other processes.
\end{abstract}

\maketitle

\section{Introduction}

Consensus, a task in which multiple processes need to agree on some value, based
on local inputs, is a fundamental problem in distributed computing. 
At the heart of this problem lies the question of whether and how it is
possible for the processes to exchange enough information with each other in
order to reach agreement, e.g., on a numerical value or on performing a joint
action. While consensus has been studied intensively for several decades already, in many models of distributed computing, it is still unknown whether and how quickly consensus can be achieved.

This paper studies deterministic consensus in dynamic directed networks. The study of such networks is of both practical and theoretical interest. 
It is of \emph{practical} relevance as the communication topology of many large-scale distributed systems is \emph{dynamic} (e.g., due to mobility, interference, or failures) and 
its links often \emph{asymmetric} (e.g., in optical or in wireless networks)~\cite{NKYG07}.
It is also of fundamental \emph{theoretical} interest, as solving consensus in dynamic directed networks is considered significantly more difficult~\cite{SWS16:ICDCN,SWK09} than solving consensus in dynamic networks with bidirectional links~\cite{KOM11}. 

We consider a worst-case perspective and assume that the information flow between the processes is
controlled by an adversary. In particular, we study a lock-step synchronous model, where  a \emph{message adversary}~\cite{AG13} may drop an arbitrary set of messages sent by some processes in each round. This results in a sequence of
directed communication graphs, whose edges tell which process can successfully 
send a message to which other process in a given round. 
We specifically consider the fundamental oblivious message
adversary model introduced by
Coulouma, Godard and Peters~\cite{CGP15}. 
In this model, the adversary is represented by a set $\bD$ of allowed
communication graphs, from which the adversary can pick one
arbitrarily in each round. 

The oblivious message adversary model is appealing because it is conceptually
simple and still provides a highly dynamic network model:
The set of allowed graphs can be arbitrary, and the
nodes that can communicate with one another can vary greatly from one round 
to the next.
It is hence also well-suited for settings where significant
transient message loss occurs, such as in wireless networks subject to
interference.
Furthermore, this model includes as a special case the classic link failure model by Santoro and Widmayer~\cite{SW89}, where up to $f$ links may fail in each round: 
the model is equivalent to a set of allowed graphs which contains all communication graphs where $\le f$ edges are missing.

Interestingly, determining
consensus solvability for a given set of graphs $\mathbf{D}$ and, in particular, designing a consensus algorithm which succeeds whenever this is possible, is difficult~\cite{CGP15}. 
For example, sometimes a ``weaker adversary'', i.e., an adversary that allows for more communication overall (e.g., supporting a larger
set $\mathbf{D}$ and failing less links), may render consensus impossible, while it is possible for a smaller set
$\mathbf{D}$.

In this paper, we are primarily interested in the \emph{time complexity} of consensus under oblivious message adversaries. Our work hence complements previous work, which either primarily focuses on the \emph{feasibility} of consensus~\cite{CGP15} or the simpler broadcast problem~\cite{ZSS18:DAM,FNW20:DAM}: how long it takes until the input value of some process has reached every other process.

\subsection{Our Contributions}

We consider the fundamental problem of distributed consensus in dynamic directed networks. 
In particular, we chart a landscape of the time complexity of consensus in the presence of oblivious message adversaries. 

Our main technical contribution is an explicit decision procedure for deciding the solvability of deterministic consensus and its analysis. This allows us, for the first time, to study the time complexity of distributed consensus under oblivious message adversaries. In particular, we present time complexity upper bounds for consensus solvability both for a centralized decision procedure as well as for solving distributed consensus. 
We further complement these upper bounds with time complexity lower bounds.

Our results also shed an interesting new light on the relationship between distributed consensus and broadcast: as the input value of some process is known to reach all other processes in almost linear time under any oblivious
message adversary~\cite{FNW20:DAM}, one might be tempted to expect that consensus 
solvability can also be decided fast. Our results show that, quite on the contrary, reaching consensus can take exponential time.

\subsection{Related Work}

Consensus is a fundamental task in distributed computing, and the question if and when consensus is possible has fascinated researchers at least since the influential impossibility result by Fischer, Lynch, and Paterson~\cite{FLP85} and its generalizations~\cite{biran1990combinatorial}. 
Consensus problems come in different flavors and arise in many settings, including shared memory architectures, message-passing systems, and 
blockchains, among others~\cite{ongaro2014search,KO11:SIGACT,CastanedaFPRRT19topo,WS19:EATCS,abraham2017blockchain}.

Research on deterministic consensus in synchronous message-passing systems 
subject to link failures dates back to the seminal 
paper by Santoro and Widmayer~\cite{SW89}, 
who showed that consensus is impossible if up to $n-1$ messages may be lost
each round. This result has later
been generalized along many dimensions~\cite{SW07,SWK09,CBS09,BSW11:hyb,CGP15,CFN15:ICALP,FNS18:PODC}.
For example, in~\cite{SWK09}, Schmid et al.~showed that consensus can even be solved when a quadratic number of messages is lost per round, provided these losses do not isolate the processes. Several generalized models have been proposed in the literature~\cite{Gaf98,KS06,CBS09}, like the heard-of model by Charron-Bost and Schiper~\cite{CBS09}, and also different agreement problems like approximate and asymptotic consensus have been studied in these models~\cite{CFN15:ICALP,FNS18:PODC}. 
In many of these and similar works on consensus~\cite{FG11,BRS12:sirocco,SWS16:ICDCN,BRSSW18:TCS,WSS19:DC,NSW19:PODC,CastanedaFPRRT19topo}, a model is considered in which, in each round, a digraph is picked from a set of possible communication graphs. Afek and Gafni coined the term message adversary 
for this abstraction~\cite{AG13}, and used it for relating problems solvable in
wait-free read-write shared memory systems to those solvable in message-passing systems.
For a detailed overview of the field, we refer to the recent survey by Winkler and Schmid~\cite{WS19:EATCS}. 
%

An interesting alternative model for dynamic networks assumes a 
$T$-interval connectivity guarantee, that is, 
a common subgraph in the communication graphs
of every $T$ consecutive rounds~\cite{KLO10:STOC,KOM11}. 
In contrast to our directional model, solving consensus is relatively
simple here, since the $T$-interval connectivity model relies on 
bidirectional links and always connected communication graphs.
For example, $1$-interval-connectivity, the weakest form of
$T$-interval connectivity, implies that all nodes are able to reach
all the other nodes in the system.

%

Another related model arises in the context of wait-free computation in shared memory systems with immediate atomic snapshots.
Roughly speaking, these systems can be described using one specific oblivious message adversary, containing all transitively closed tournaments.
Wait-free computation in this context is often studied using topological tools~\cite{HKR13,AttiyaC13,AttiyaCHP19,Kozlov15,Kozlov16}. This line of work did not provide any time complexity bounds for consensus in our model, however.

Closely related to our work is the paper by
Coulouma, Godard, and Peters~\cite{CGP15}, who substantially refined the results of~\cite{SW07}. The authors consider oblivious message adversaries and identify an equivalence relation on the sets of communication graphs, which captures the essence of consensus impossibility via non-broadcastability of one of the equivalence classes (``beta classes'') of this relation. The paper also presents a distributed consensus algorithm that, essentially, computes the beta classes. However, in contrast to our paper, the main focus of this work is on feasibility of consensus.

To the best of our knowledge, we are the first to provide an efficient (centralized) decision procedure and a distributed consensus algorithm with worst-case time complexity guarantees under oblivious message adversaries.

\subsection{Organization}

The remainder of this paper is organized as follows.
We introduce our formal model and terminology in Section~\ref{sec:notation}.
The description and analysis of our decision procedure and our consensus 
algorithm are presented in Section~\ref{sec:decProc} and Section~\ref{sec:consensus}, respectively, and our lower bound results are presented in Section~\ref{sec:lowerbound}. 
We conclude our contribution and discuss directions for future work in Section~\ref{sec:futwork}. Due to space constraints, most proofs and additional findings are deferred to the appendix.

\section{Model and Preliminaries}\label{sec:notation}

We assume a set $\Pi = \{ p_1, \ldots, p_n\}$ of $n$ processes, which execute
a deterministic distributed protocol to reach consensus.
Processes operate in lock-step synchronous rounds, where each round consists of a phase of message exchanges among the processes, followed by some local computation, whose execution time is assumed to be negligible. We consider a \emph{full information} protocol where, in each round, every process broadcasts its complete local history (its 
\emph{view} obtained at the end of the previous round, or the initial state), and computes a deterministic \emph{decision function} $\Delta$ based on its current view, which also involves all views it received from other processes in this round.

Each phase of message exchange is restricted by a (possibly different) directed graph on $\Pi$, called a \emph{communication graph}, which is controlled by a message adversary.
A message from $p$ to $q$ may be delivered in round $r$ only if the
communication graph of round $r$ contains the edge $(p, q)$. Since
every process obviously knows its own current view, we just assume that the communication graph always contains all the self-loops.
We use $\In_G(v)$ denote the in-neighborhood of process $v$ in a graph $G$.
Messages are unacknowledged and
rounds are communication-closed, i.e., messages that are sent in round $r$
arrive in round $r$ or not at all.

A \emph{communication pattern} is a sequence of such communication graphs, which (along
with the initial views of all processes and the decision function $\Delta$)
will uniquely define a run of the system.
In the oblivious message adversary model, there is a set $\bD$ of allowed
communication graphs,
and the admissible communication patterns are all sequences of graphs from $\bD$.
For brevity, we identify our message adversary with its set of allowed communication graphs.

For a communication graph $G$, let $G^r = (G)_{i=1}^r$
denote the communication pattern that consists of $r$ repetitions of $G$.
For a set of communication graphs $\bG$, let
$\bG^r = \{ (G_i)_{i=1}^r : G_i \in \bG \}$
be the set of communication patterns of length $r$ that consist
only of graphs from $\bG$.
Given a set of allowed graphs $\bD$, the oblivious message adversary
generated by $\bD$ may thus be written as $\bD^\omega$ ($\omega$
denotes infinitely many repetitions of elements of $\bD$).

Let $\sigma = (G_i)_{i=1}^r$ be a communication pattern, where its length $r \ge 1$ can be
any integer or infinite (denoted~$\omega$), and let $\Sigma$ be a set of communication patterns.
We use $\sigma|_{r'} = (G_i)_{i=1}^{r'}$ to denote the $r'$-round prefix of $\sigma$, which is
only defined if the length of $\sigma$ is at least $r'$, 
and $\Sigma|_{r'} = \{ \sigma|_{r'} : \sigma \in \Sigma \}$ to denote the set of all $r'$-round prefixes of $\Sigma$; by convention, $\sigma|_0 = \varepsilon$, where $\varepsilon$ is the empty word. 
We use $\sigma(r') = G_{r'}$ to denote the $r'$\textsuperscript{th} graph of $\sigma$
and $\Sigma(r') = \{ \sigma(r') : \sigma \in \Sigma \}$ for the set of communication patterns $\Sigma$.
If $\sigma$ has a finite length $r$ and $H$ is an arbitrary communication graph,  we write $\sigma' = \sigma \circ H$ to denote $\sigma$ extended by $H$, i.e.,
the communication pattern of length $r+1$ with 
$\sigma'(i) = \sigma(i)$ for all $i \le r$ and
$\sigma'(r+1) = H$.

A \emph{root component} of a graph is a strongly connected component that has no incoming edge from a node outside of the component.
We call a graph $G$ \emph{rooted} if it has a single root component and
write $\Root(G)$ for the node set of the root component of $G$.
Note that if a graph $G$ is rooted then a node (in our context: a process) $p \in V(G)$ has a path to every other node (process) in $G$ if and only if $p \in \Root(G)$.
In \cref{claim:non-rooted} below, we show that consensus is trivially impossible if the set of allowed graphs contains a graph that is not rooted, and
for this reason
we consider adversaries whose set $\bD$ consists of rooted graphs only.
A set of communication graphs $\bS$ is \emph{root-compatible} if all their root components 
contain a common node, i.e.,
$\bigcap_{G \in \bS} \Root(G) \ne \emptyset$.
We will show that root-compatibility is a central concept when it comes to consensus solvability.

%
In our full information protocol, the view of process $p$ in $\sigma$ at time 
(= end of round) $r\geq 1$ comprises the view of all the processes that $p$ had in its in-neighborhood in 
the round $r$ communication graph $\sigma(r)$, along with the round
number $r$. The initial view of process $p$ consists of its input value $x_p$ (see the specification
of the consensus problem below) and the round number 0. Formally, views are
recursively defined as
$\view_\sigma(p, 0) = \{ (p, 0, x_p) \}$ and, for $r>0$,
$\view_\sigma(p, r) = (p, r, V_\sigma(p, r-1))$, where
$V_\sigma(p, r-1) = \{ \view_\sigma(q, r-1) : (q, p) \in \sigma(r) \}$. 

For notational simplicity, we will subsequently use the tuple $(p,r)$, 
called a \emph{process-time node}, to refer to the view of process $p$ at 
time $r$.
We thus use $(p, r') \leadsto_\sigma (q, r)$ to denote that $p$ at time $r'<r$ has influenced $q$ at time $r$, which can be expressed formally by the existence of a sequence of processes $p=p_1,\dots,p_{r-r'+1}=q$ satisfying $\view_\sigma(p_i, r'+i-1) \in V_\sigma(p_{i+1}, r'+i-1)$ for $1 \leq i \leq r-r'$. We say that $p$ is a broadcaster in $\sigma$ (or equivalently, that 
a communication pattern $\sigma$ is \emph{broadcastable} by $p$),
if $(p, 0) \leadsto_\sigma (q, r)$ for some time $r$, for all~$q \in \Pi$.

Two communication patterns $\sigma$ and $\sigma'$ of the same length
are \emph{indistinguishable} by a process $p$, denoted 
$\sigma\sim_p\sigma'$,
if this process has the same view in $\sigma$ and in $\sigma'$, \emph{eventually} or \emph{in each round} $r$ in case of infinite patterns. 
Formally, 
$\sigma \sim_p \sigma'\Leftrightarrow \view_\sigma(p, r) = \view_{\sigma'}(p, r)$ if $\sigma$ and $\sigma'$ are $r$-round patterns, and
$\sigma \sim_p \sigma'\Leftrightarrow \view_\sigma(p, r) = \view_{\sigma'}(p, r)$ for all $r$ if $\sigma$ and $\sigma'$ are infinite.
We write $\sigma \sim \sigma'$ if $\sigma \sim_p \sigma'$ for some $p$.
We also use $\sigma \not\sim_p \sigma'\Leftrightarrow \neg (\sigma \sim_p \sigma')$, and
$\sigma \not\sim \sigma'\Leftrightarrow (\forall p\in\Pi:\sigma \not\sim_p \sigma')$.

Given a set $\Sigma$ of communication patterns of the same length, we define its \emph{indistinguishability graph} $I(\Sigma)$ as follows.
The nodes of $I(\Sigma)$ are the communication patterns in $\Sigma$, and the two communication patterns $\sigma, \sigma' \in \Sigma$ are connected by an edge if  $\sigma \sim \sigma'$, i.e., if they are indistinguishable for some process.
We label each edge with the set of processes defining it, that is, we define an edge labeling function $\ell:E(I(\Sigma))\to 2^\Pi$ by 
$\ell((\sigma,\sigma')) = \{ p \in \Pi : \sigma \sim_p \sigma' \}$.

Our first simple, yet important insight is that root components can preserve
indistinguishability. Consider two communication patterns $\sigma, \sigma'$ 
that are indistinguishable for a set of processes $\ell((\sigma,\sigma'))$,
and assume that there is an allowed graph $G\in \bD$ such that $\Root(G)\subseteq\ell((\sigma,\sigma'))$.
Then, the communication patterns $\sigma \circ G$ and $\sigma' \circ G$ are also indistinguishable for the processes in $\Root(G)$: in $G$, these processes
only receive messages from other members of $\Root(G)$, and so these extended communication patterns are still indistinguishable for them.

\begin{claim}
  \label{claim:ignoramus}
  Let $\bD$ be an oblivious message adversary, $r$ be a round, and
  $e = (\sigma, \sigma')$ be an edge in $I(\bD^r)$. For $r>1$, the edge
$(\sigma|_{r-1},\sigma'|_{r-1})$ is in $I(\bD^{r-1})$. Moreover,  
  if there is a graph $G \in \bD$ such that $\Root(G) \subseteq \ell(e)$
  then the edge $e' = (\sigma \circ G, \sigma' \circ G)$ is in $I(\bD^{r+1})$ and its label $\ell(e')$ satisfies 
  $\Root(G) \subseteq \ell(e') \subseteq \ell(e)$.
\end{claim}

\begin{proof}
If $r>0$, for every $p\in\ell(e)$, the indistinguishability $\sigma \sim_p \sigma'$ also
implies $\sigma|_{r-1} \sim_p \sigma'|_{r-1}$, so the edge
$(\sigma|_{r-1},\sigma'|_{r-1})$ is indeed in $I(\bD^{r-1})$.

To prove the second part of our claim, consider any process $p \in \Root(G)$.
  By the definition of a root component, we have $\In_G(p) \subseteq \Root(p)$, 
  so each process $q$ with $(q, r) \in \view_{\sigma \circ G}(p, r+1)$,
  is in $\Root(G)$, and satisfies $\view_{\sigma}(q, r) = \view_{\sigma'}(q, r)$,
	because $\Root(G) \subseteq \ell(e)$.
  This immediately implies that 
  $\view_{\sigma \circ G}(p, r+1) = \view_{\sigma' \circ G}(p, r+1)$
  and thus the edge $e'$ exists and
  $\Root(G) \subseteq \ell(e')$.
	The last part, $\ell(e') \subseteq \ell(e)$, follows because if
	$\view_{\sigma \circ G}(q, r+1) = \view_{\sigma' \circ G}(q, r+1) = (q, r+1, V_\sigma(q,r))$ for some process
	$q$ then
	$\view_{\sigma}(q, r) = \view_{\sigma}(q, r)$,
	as, by definition, $\view_\sigma(q, r)\in V_\sigma(q, r)$.
\end{proof}

In the \textbf{\emph{consensus problem}}, 
each process $p$ has an input value $x_p \in V$, taken from some finite domain $V$, and an
output value $y_p$, initialized to $\bot$, to which it can write irrevocably, i.e., only once.
An algorithm solves consensus in our setting if it ensures that
\begin{itemize}
  \item eventually, every process $p$ decides, i.e., assigns $y_p \ne \bot$ (termination),
  \item if $y_p \ne \bot$ and $y_q \ne \bot$ then $y_p = y_q$ for all $p, q \in \Pi$ (agreement),
  \item if $y_p = v \ne \bot$ then there is a process $q \in \Pi$ such that
  $x_q = v$ (validity).
\end{itemize}

Since we will consider full information protocols only, our consensus algorithm
is actually a collection of decision functions. For every $p \in \Pi$, the 
decision function $\Delta_p$
maps every possible $\view_\sigma(p,r)$ to a decision value $y_p \in V \cup \{\bot\}$,
such that $\Delta(\view_\sigma(p,r))\neq\bot$ implies
$\Delta(\view_\sigma(p,r'))=\Delta(\view_\sigma(p,r))$ for every $r'\geq r$.
The \emph{configuration} $C_\sigma^r$ of our system at the end of round $r$ in $\sigma$, 
is the vector of the elements $(\view_\sigma(p,r),\Delta(\view_\sigma(p,r)))$, 
for all $p$,  
and the \emph{run} (also called execution in the literature) corresponding to $\sigma$ is the sequence $(C_\sigma^r)_{r\geq 0}$.
In the oblivious message adversary model, a run
is uniquely determined by the input value assignment contained in the initial views 
and the communication pattern since the algorithm is deterministic.

With these definitions in mind, we now state two properties  
of consensus under oblivious message-adversaries, which 
will be of central importance
in this paper. 
We first observe that any valid decision value must be the input value of a
broadcaster.
The proof of the following claim uses the same argument as~\cite[Theorem 2]{WSM19:OPODIS}.

\begin{claim}
	Let $\bD$ be an oblivious message adversary and let $\sigma \in \bD^\omega$.
	If in some correct consensus algorithm, all processes decide $v$ in a run
	with $\sigma$, then $v$ is the input value of a broadcaster in $\sigma$.
  \label{claim:validity}
\end{claim}

\begin{proof}
	By the termination condition, there is a round $r$ such that in all runs with
	$\sigma$ all processes decide by this round when running a given correct
	consensus algorithm.
	Suppose that there is a $r$-round run $\varepsilon$ with communication pattern
	$\sigma$ where all processes decide $v$ even though no broadcaster in
	$\sigma$ has input value $v$.
	We show that this leads to a contradiction to the assumed correctness of the
	consensus algorithm.

	Let $P=\{i_1, \ldots, i_k \}$ be the identifiers of those processes that
	start with input value $v$ in $\varepsilon$.
	By the valditiy condition, $P \ne \emptyset$.
	Let $\varepsilon_j$ denote the run that is the same as $\varepsilon$,
	except that the processes with identifiers $i_1, \ldots, i_j$ have an input value $\ne v$.
	We show by induction that some process decides $v$ in $\varepsilon_j$ for
	$0 \le j \le k$.
	Thus in the run $\varepsilon_k$ some process decides $v$, even though no
	process has input $v$ in this run, a contradiction to the validity condition
	of consensus.

	The base of the induction $j=0$ follows immediately because
	$\varepsilon \sim \varepsilon_0 = \varepsilon$.

	For the step from $j$ to $j+1$, where $0 \le j < k$,
	we observe that, because $\sigma$ is not broadcastable for any process with
	an identifier from $P$,
	there is a process $q$ such that $(p_{i_{j+1}}, 0) \not\leadsto (q, r)$.
	Since
	$\varepsilon_j$ is identical to
	$\varepsilon_{j+1}$ except for the input of $p_{i_{j+1}}$,
	we have
	$\varepsilon_j \sim_q
	\varepsilon_{j+1}$.
	As all processes decide by round $r$ in
	$\varepsilon_j$,
	and because they decide $v$ by hypothesis,
	$q$ and, by agreement, all processes decide $v$ in $\varepsilon_{j+1}$.
\end{proof}

Our second observation is that every communication graph in the set of allowed graphs
of an oblivious message adversary, under which consensus is solvable, must be rooted.

\begin{claim}
  If an oblivious message adversary contains, in its set of allowed graphs $\bD$, a graph $G$ that is not rooted, then consensus is impossible.
  \label{claim:non-rooted}
\end{claim}

\begin{proof}
  The pattern $\sigma=G^\omega$ may be played by the adversary even though it is not
  broadcastable by any process, thus the claim follows from \cref{claim:validity}.
\end{proof}

\section{A Decision Procedure for Consensus Solvability}
\label{sec:decProc}

In this section, we present a decision procedure for determining whether consensus is solvable under an oblivious message adversary with a set $\bD$ of allowed graphs. 
In a nutshell, our procedure revolves around the (undirected) indistinguishability graph 
$I(\bD)$, constructed from the given input set $\bD$:
the nodes of the indistinguishability graph represent the graphs of $\bD$ and 
the edges represent indistinguishability. 
Given $I(\bD)$, we create a sequence $\N_1 = I(\bD),\N_2,\ldots $
of refinements of $I(\bD)$,
and use the last graph $\N_{\TD}$ to decide if consensus is solvable under the message adversary $\bD$. Here, $\TD$ is the number of iterations of the decision procedure, that is, the time complexity of the algorithm.
In some sense, our decision procedure can essentially 
be viewed as an explicit computation of the abstract beta classes (and their broadcastability), as introduced by Couloma et al.~\cite{CGP15}.
As an additional feature, it reveals a crucial and previously unknown relation between the number of iterations of the
decision procedure under a given oblivious message adversary and the time
complexity of distributed consensus.

More concretely, our approach, summarized in \cref{alg:Nerve}, uses the fact that a graph whose root component is a subset of $\ell(e)$ is suitable for perpetuating the indistinguishability for at least some of the processes of $\ell(e)$ (according to
\cref{claim:ignoramus}).
The algorithms starts from the indistinguishability graph $\N_1=I(\bD)$ of $\bD$, where $\bD$ is viewed as a set of $1$-round communication patterns:
the nodes of $I(\bD)$ are the graphs of $\bD$,
and two graphs 
$G, G' \in \bD$ are connected by an edge if there is a process $p$ that has the same set of incoming edges in $G$ and in $G'$.
The algorithm then computes a sequence $(\N_i)$ of graphs, using iterative refinement. To refine from $\N_{i-1}$ to $\N_i$, it keeps all $\N_{i-1}$'s nodes, but only a subset of its edges (Line~\ref{line:rm}):
an edge $e=(u,v)$ is kept (by adding it to the set $E_i$) if
the connected component of $e$ in $\N_{i-1}$ contains a communication graph $G$ such that $\Root(G)\subseteq\ell((u,v))$ (Line~\ref{line:guardRM}).

This procedure continues until the set of edges does not change for
two successive iterations, or until
all remaining connected components are root-compatible, i.e., all its
communication graphs have a common member in their respective root
components.
As we will see later in \cref{thm:imposs,thm:poss},
the root-compatibility of the connected
components of the refined indistinguishability graph is precisely what is required to make consensus solvable.

For the algorithm,
we assume that all graphs of $\bD$ have a unique root component,
as consensus is trivially impossible otherwise (\cref{claim:non-rooted}).
Note that, for two communication graphs $G, H$,
we have $\ell((G,H)) = \{ p \in \Pi : G \sim_p H \} = \{ p \in \Pi : \In_G(p) = \In_H(p) \}$.

\begin{wrapfigure}[23]{R}{0.5\textwidth}
		\vspace*{-3ex}
\center
\begin{minipage}{0.5\textwidth}
		\small
\begin{algorithm}[H]
	\DontPrintSemicolon
	\KwIn{A set of allowed graphs $\bD$}
	\KwOut{The refined indistinguishability graph $\N_{\TD}$. Consensus is solvable if and only if all connected components of $\N_{\TD}$ are root-compatible.}
	\BlankLine
	
	\tcp{Initialization:}
	$i \gets 1$ \;
	$\N_1 \gets I(\bD)$ \;
	\BlankLine
	
	\tcp{Iterative construciton:}
	\Repeat{$\N_i = \N_{i-1}$ or all connected components of $\N_i$ are root-compatible \label{line:until}}{
		$i \gets i+1$ \;
		$E_i \gets \emptyset$ \;
		\ForEach{$e \in E_{i-1}$}{ \label{line:DecProcIterations}
			Let $\bG$ be the communication graphs reachable from $e$ in $\N_{i-1}$ \;
			\If{$\exists G \in \bG : \Root(G)\subseteq\ell(e)$ \label{line:guardRM}}{
				$E_i \gets E_i \cup\{e\}$ \label{line:rm}
			}
		}
		$\N_i \gets \langle \bD, E_i \rangle$\;
	}
  \Return $\N_{i-1}$
\end{algorithm}
\caption{The consensus decision procedure. It iteratively constructs the refined indistinguishability graph $\N_{\TD}$ for a set of allowed graphs $\bD$.}
\label[algorithm]{alg:Nerve}
\end{minipage}
\end{wrapfigure}

The following corollary provides a concise statement of the rule according to which the decision procedure selects which edges to keep when refining $\N_{i-1}(\bD)$ into $\N_i(\bD)$.

\begin{corollary}
  \label{cor:nerveEdge}
  Let $e = (A, B)$ be an edge of $\N_i(\bD)$, for $i > 1$.
  Then in $\N_{i-1}(\bD)$:
  \begin{enumerate}
  	\item 
  	the edge $e = (A, B)$ exists, and
  	\item
	there exists a node $G_e$ with $\Root(G_e) \subseteq \ell(e)$, such that $A,B$ and $G_e$ are in the same connected component.
    \end{enumerate}
\end{corollary}

\begin{proof}
  According to \cref{alg:Nerve}, an edge $e = (A, B)$ can only
  persist in $\N_i$ if it was already present in $\N_{i-1}$ and there was a
  corresponding graph $G_e$ with $\Root(G_e) \subseteq \ell(e)$ connected to $A$
  and $B$ in $\N_{i-1}$.
\end{proof}

We observe that, in order for an edge $e$
of the indistinguishability graph to be ``protected'' from being
omitted by the decision procedure by Line~\ref{line:rm} of \cref{alg:Nerve},
there must exist a communication graph whose root component is a subset of the label of $e$.
This motivates the following definition.

\begin{definition}
  \label{def:protected}
  Given a set of allowed graphs $\bD$,
  let $E$ be a set of edges of $I(\bD)$ and $\bG \subseteq \bD$ be a set of
  communication graphs.
  We call $E$ \emph{protected} by $\bG$ if for every $e \in E$ there is a graph
  $G_e \in \bG$ such that $\Root(G_e) \subseteq \ell(e)$.
\end{definition}

The following upper bound on the number of iterations $\TD$ of the decision procedure
exploits the maximum number of different labels of the edges of $I(\bD)$.

\begin{claim}
	The number of iterations of the decision procedure, $\TD$, satisfies $\TD \le 2^n$.\label{claim:decisionupper}
\end{claim}

\begin{proof}
	For a set of communication graphs $\bG$, let $\N_i[\bG]$ denote the subgraph
	of $\N_i$ induced by $\bG$.
	According to \cref{alg:Nerve}, there must exist a set of communication graphs
	$\bG \subseteq \bD$ such that $\N_{i}[\bG]$ is connected and not
	root-compatible for all $i < \TD$, whereas all connected components of
	$\N_{\TD}$ are root-compatible.
	That is, $\bG$ constitutes the last connected component of $I(\bD)$ that had
	to be broken apart by the decision procedure in order to arrive at a graph
	$\N_{\TD}$ where all connected components are root-compatible.

	Furthermore, for $1 < i < \TD$, the set $\bC_i(\bG)$ of nodes reachable from 
	$\bG$ in $\N_i$ satisfies $|\bC_i(\bG)| < |\bC_{i-1}(\bG)|$.
	This is because, if the $(i-1)$\textsuperscript{th} iteration of the decision
	procedure does not result in the removal of a node from $\bC_{i-1}(\bG)$,
	then a set of edges that connects $\bC_{i-1}(\bG)$ in $\N_{i-1}$ is protected by
	the communication graphs of $\bC_{i-1}$; hence, no
	node will be removed from $\bC_j(\bG)$ for any $j \ge i$.
	This cannot come to pass, however, because then the decision procedure would
	already have terminated after $i < \TD$ iterations.

	In addition, all edges $e$ of the connected component of $\bG$ in $\N_i$ that
	have the same label $\ell(e) = \lambda$ are removed during a single iteration
	of the decision procedure:
	If $e$ is removed from the connected component of $\bG$ in $\N_i$, then there
	is no communication graph in $\bC_i(\bG)$ that protects $e$ and so all edges
	with label $\lambda$ are removed from the connected component of $\bG$.
	We recall that every label is a nonempty subset of $\Pi$, thus there are at most $2^n-1$
	different labels.
	The claim follows because, as we have shown above,
	$|\bC_i(\bG)| < |\bC_{i-1}(\bG)|$;
	hence at least one edge is removed from the connected component of $\bG$ in
	$\N_i$ during the $i$\textsuperscript{th} iteration of the decision
	procedure.
\end{proof}

Before looking more closely into the ramifications of a large number of
iterations $\TD$ of the decision procedure of a given oblivious message
adversary $\bD$, it is instructive to study 
a few ``extreme'' examples of such adversaries, and, in particular, how the number
of communication graphs $|\bD|$ relates to $\TD$.
First, one may wonder whether the decision procedure can be fast if the set $\bD$ of allowed graphs is exponentially large.
An example for such a scenario, in which consensus is solvable, is the set of all communication graphs that consist of a single clique
of a fixed size $\lfloor n/c \rfloor$, for a constant $c$,
and all the edges from each clique node to all other nodes
(plus the self loops).
There are exponentially many such graphs, yet no two are indistinguishable to any of the nodes, 
so the decision procedure already terminates after the first iteration because all connected components in $I(\bD)$ consist of a single communication graph.
An example where a fast decision is possible despite an exponentially sized $\bD$, where consensus is impossible, is the set of all rooted trees for $n>2$.
In this case, there is a path in $I(\bD)$ connecting every two trees $T_1, T_2$.
Also, every edge $e$ in $I(\bD)$ has a corresponding tree $T \in \bD$ that
protects this edge, since there is a tree $T$ with $\Root(T) \subseteq \ell(e)$.

Complementing these insights, the question arises whether there are examples where
$\TD$ is (almost) the same as $|\bD|$. We will 
answer this question affirmatively (in \cref{sec:lowerbound}), by giving an explicit example where
$\TD$ is even exponential in $n$. In a nutshell, we will choose a set of communication graphs
$\bD = \{ G_1, \ldots, G_{\TD} \}$,
where the root component of each graph consists of a different set of processes
of the same cardinality, i.e., for every $G, G' \in \bD$ we have
$|\Root(G)| = |\Root(G')|$, but if $G\neq G'$ then $\Root(G) \neq \Root(G')$.
Furthermore, we let
\begin{equation}
	G_1 \sim_{R_3} G_2 \sim_{R_4} \ldots \sim_{R_{\TD}} G_{\TD-1} \sim_S G_{\TD},\label{eq:examplechain}
\end{equation}
where $R_i = \Root(G_i)$ 
and $S$ is a nonempty set such that no $G \in \bD$ satisfies
$\Root(G) \subseteq S$.
Here, the decision procedure can remove only the rightmost edge $\sim_S$ in the first iteration,
only the edge $\sim_{R_{\TD}}$ in the second iteration, and so on,
because all the remaining edges are protected by one of the remaining graphs.

Also in this case, consensus might be solvable (as in the example in
\cref{sec:lowerbound} described above), or it might be impossible, as in the instance
\[
	G'_1 \sim_{R'_3} G'_2 \sim_{R'_1} G'_3 = G_1 \sim_{R_3} G_2 \sim_{R_4} \ldots
	\sim_{R_{\TD}} G_{\TD-1} \sim_S G_{\TD}
\]
where we assume that $G'_1$ and $G'_2$ are chosen such that they are not root-compatible: in
this case, the indistinguishability $G'_1 \sim_{R'_3} G'_2$ will never break.

In view of the above results, it might be tempting to assume that $\TD$ also
determines the termination time of distributed consensus. Interestingly, this is not the case. Complementing the result of \cref{thm:poss} established in \cref{sec:consensus}, we will show in \cref{sec:source} that there are instances 
of oblivious message adversaries where the decision
procedure terminates after a constant number of iterations, while the consensus
terminatino time is exponential in $n$.


\section{Time Complexity of Consensus}\label{sec:consensus}

In this section, we study the time complexity of consensus,
and also ascertain our claim from \cref{sec:decProc}, namely, that
the decision procedure of \cref{alg:Nerve} correctly assesses oblivious message
adversaries where consensus is solvable.
Thus, throughout this section, we consider an oblivious message adversary, where, after some number $\TD$ of iterations, 
\cref{alg:Nerve} determined that all connected components of the refined
indistinguishability graph $\N_{\TD}$ are root-compatible.

	For solving consensus, we use the fact that non-connectivity in $\N_{\TD}$
	implies non-connectivity in $I(\bD^{(n-1)\TD+1})$, in the following sense:
	Let $\bC_1$ and $\bC_2$ be two different connected components of $\N_{\TD}$, and $t > (n-1) \TD$.
	Then, any two communication patterns 
	$\sigma_1\in\bC_1^t$ and $\sigma_2\in\bC_2^t$, consisting only of graphs of $\bC_1$ and 
$\bC_2$, respectively, 
	are not connected in the indistinguishability graph $I(\bD^t)$.

We then apply a pigeon-hole argument to show that all connected components of
$I(\bD^{c t})$ are broadcastable, where $c$ is the number of connected
components of $\N_{\TD}$. Note that this choice guarantees that graphs from at least one 
connected component are used at least $t$ times. 
From here, a consensus decision function $\Delta_p$
can be easily defined by 
(i) for each connected component $\cC$ of $I(\bD^{c t})$, choosing one of its broadcasters, denoted $b(\cC)$, and 
(ii) 
if $p$'s view is consistent with a graph sequence $\sigma$,
and $\sigma$ belongs to a connected component $\cC$ of $I(\bD^{c t})$,
then $p$ decides on the input $x_{b(\cC)}$ of $b(\cC)$, 
for which $\view_\sigma(b(\cC),0,x_{b(\cC)})$ must already be present in $p$'s view.
	
It is rather immediate that such a procedure solves consensus, given the mapping
$b(\cC)$, which we will prove in the remainder of this section:
Termination follows from the existence of the mapping $b(\cC)$; 
validity follows because the decided value was some process' input value; 
agreement is a consequence of all pairwise indistinguishable views lying in
the same connected component $\cC$ of $I(\bD^{c t})$.
Hence two different decisions can only occur in runs that are distinguishable
for everyone (and are thus distinct runs).

A path $\pi=(\sigma_0,\ldots,\sigma_s)$ in $I(\bD^r)$
is a sequence of communication patterns such that $(\sigma_i,\sigma_{i+1})\in E(I(\bD^r))$ 
for all $0\leq i<s$.
Given such a path and $r'\leq r$, we write $\pi|_{r'}$ to denote the path $\left(\sigma_0|_{r'},\ldots,\sigma_\ell|_{r'}\right)$ in $I(\bD^{r'})$
of the ${r'}$-round prefixes of the communication patterns in $\pi$, which exists by 
\cref{claim:ignoramus}.
Similarly, we denote by $\pi(r')$ the path $(\sigma_0({r'}),\ldots,\sigma_\ell({r'}))$ in $I(\bD)$ of
the ${r'}$\textsuperscript{th} graphs of the communication patterns in $\pi$.
Both $\pi|_{r'}$ and $\pi(r')$ are indeed paths in the corresponding indistinguishability graphs, due to a more general claim: removing an intermediate communication round from all communication patterns in a path cannot disconnect it, as stated below.

For a communication pattern $\sigma$ of length $r$, and some round $r' \le r$, let $\sigma - r'$ denote $\sigma|_{r'-1} \circ \sigma(r'+1) \circ \cdots \circ \sigma(r)$, i.e., the communication pattern $\sigma$ with the round $r'$ communication graph omitted.
\cref{cor:removeRound} shows that edges, and hence paths, between communication patterns in $I(\bD^r)$ 
are preserved when omitting some round $r'$.

\begin{corollary}
	\label{cor:removeRound}
  If the edge $(\sigma,\sigma')$ is in $I(\bD^r)$, then the edge $(\sigma - r',
  \sigma' - r')$ is in $I(\bD^{r-1})$ as well.
\end{corollary}

\begin{proof}
	Assume for contradiction that the edge is not preserved, i.e., $\sigma \sim \sigma'$ while $\sigma - r' \not\sim \sigma'-r'$.
	So, there is a process $p$ such that 
	$\sigma \sim_p \sigma'$ (this is true for at least one process, $p$)
	while
	$\sigma - r' \not\sim_p \sigma'-r'$ (this is true for all processes, and specifically for $p$).
	This implies that there exists a round $r''\neq r'$ and a process $q$ with
        w.l.o.g.\ $(q, r'') \leadsto_{\sigma-r'} (p, r)$ but $(q, r'') \not\leadsto_{\sigma'-r'} (p, r)$
        or
	$\view_{\sigma-r'}(q,r'')\neq\view_{\sigma'-r'}(q,r'')$:
	if no such $q,r''$ existed, we would have
	$\sigma - r' \sim_p \sigma' -r'$.
	Since $(q, r'') \leadsto_{\sigma-r'} (p, r)$, 
	we also have $(q, r'') \leadsto_{\sigma} (p, r)$,
	as the sequence of processes causing $(q, r'')$ to be in  $\view_{\sigma-r'}(p, r)$ also exists in $\sigma$ and we just need to take path where the process of round $r'$ is the same as of round $r'-1$.
To finish, it suffices to consider two cases: 
	if $(q, r'') \not\leadsto_{\sigma'} (p, r)$, then $p$ distinguishes $\sigma$ and $\sigma'$ 
	since it has $\view_{\sigma}(q,r'')$ in its view in $\sigma$ but does not have $\view_{\sigma'}(q,r'')$ in its view in $\sigma'$;
	if $(q, r'') \leadsto_{\sigma'} (p, r)$, then $p$ distinguishes $\sigma$ and $\sigma'$ by having  	$\view_{\sigma}(q,r'')\neq\view_{\sigma'}(q,r'')$ in its views.
	In both cases $\sigma \not\sim_p \sigma'$, a contradiction.
\end{proof}

The following corollary relates the preservation of an edge in $I(\bD^r)$ to the root components of the communication graphs that occur in the communication patterns of this edge.

\begin{corollary}
\label{cor:infProp}
Let $\bD$ be a set of allowed graphs and $0<r'<r$ integers.
Consider an edge 
$e=(\sigma,\sigma') \in I(\bD^r)$
such that 
$e'=(\sigma|_{r'},\sigma'|_{r'}) \in I(\bD^{r'})$
satisfies
$\sigma|_{r'}\neq\sigma'|_{r'}$.
Then, there are at most $|\ell(e')|-1$ rounds $r_j$, $r'<r_j\leq r$, satisfying 
$\Root(\sigma(r_j)) \not\subseteq \ell(e')$.
\end{corollary}

\begin{proof}
  By \cref{claim:ignoramus}, we can be sure that $e'$ exists. For a contradiction,
  suppose that there are $|\ell(e')|$ rounds
  $r'< r_1<\cdots<r_{|\ell(e')|}\leq r$
  such that each $r_j$ satisfies $\Root(\sigma(r_j)) \not\subseteq \ell(e')$.
  Let 
  \begin{equation}
  U_j =
  \{ p \in \Pi : \exists q \in \Pi  \setminus \ell(e') \; 
  (q, r') \leadsto_{\sigma} (p, r_j) \} \label{eq:Ujdef}
  \end{equation}
  denote the set of processes that received a message by round $r_j$, sent after round $r'$,
  from a process outside of $\ell(e')$.
  Let $r_0=r'$ and $U_0=\Pi \setminus \ell(e')$.
  Note that from $\sigma|_{r'}\neq\sigma'|_{r'}$ it follows that
  $\emptyset \neq \ell(e') \neq \Pi$ and thus $U_0 \ne \emptyset$.
  
  Let $\overline U_j = \Pi \setminus U_j$ and consider the cut
  $(U_j, \overline U_j)$ in $\sigma(r_j)$, the communication graph at round $r_j$.
  Since we have
  $\Root(\sigma(r_j)) \not\subseteq \ell(e')$,
  there is a process $p' \in \Root(\sigma(r_j)) \setminus \ell(e')$.
  On the one hand, $p' \in \Root(\sigma(r_j)) \setminus \ell(e')$ immediately implies $p' \in U_j$, since $(p', r') \leadsto_\sigma (p', r_j)$.
  On the other hand, $p' \in \Root(\sigma(r_j))$ implies that in $\sigma(r_j)$ there is a path from $p'$ to every node.
  Hence, if $\overline U_j\neq\emptyset$, then there is a node $p''\in \overline U_j$, and a path in $\sigma(r_j)$ from $p'$ to $p''$;
  this path must cross an edge $\tilde e_j$ from $U_j$ to $\overline U_j$.
        
We now use induction on $j=0,\ldots,|\ell(e')|$ to show that $|U_j| \ge n-|\ell(e')|+j$. For the basis
$j=0$, we have already shown that $|U_0|=n-|\ell(e')| > 0$.
  In the induction step, we prove that $U_j$ grows by at least one (unless $U_j = \Pi$) due to the edge $\tilde e_j = (q', q'')$ from $U_j$ to
  $\overline U_j$.
  As, for every $q \in \Pi\setminus \ell(e')$ in the definition if $U_j$ in \cref{eq:Ujdef}, 
$(q, r') \leadsto_\sigma (q', r_j)$ in conjunction with
  $(q', r_j) \leadsto_\sigma (q'', r_{j+1})$ implies
  $(q, r') \leadsto_\sigma (q'', r_{j+1})$, we obtain 
  $U_{j+1} \supseteq U_j \cup \{ q'' \}$ as required.
  
  It hence follows that $|U_{|\ell(e')|}| =n$,
  i.e., by round $r \ge r_{|\ell(e')|}$, every process has received a message, sent after round $r'$, from a process $q$ outside of
  $\ell(e')$. Consequently, at time $r$, the view of every process contains the view of a process $q$ that could distinguish $\sigma|_{r'}$ and $\sigma'|_{r'}$, hence every process can also 
  distinguish $\sigma$ and $\sigma'$. Formally,
    $\forall p \in \Pi \: \exists q \in \Pi \setminus \ell(e) :
    (q, r')  \leadsto_{\sigma} (p,r) \text{ and }
    \view_{\sigma}(q, r') \ne \view_{\sigma'}(q, r')$,
  which implies that $\view_{\sigma}(p,r) \ne \view_{\sigma'}(p,r)$.
  That is, every process that can distinguish $\sigma|_{r'}$ and $\sigma'|_{r'}$
  can also distinguish $\sigma$ and $\sigma'$,
  contradicting the existence of the edge $e_r=(\sigma,\sigma')$ in $I(\bD^r)$.
\end{proof}

We proceed with \cref{lem:protectedComponents}, which generalizes and formalizes
chains like \cref{eq:examplechain}, made up of connected subgraphs $\cS_1,\dots,\cS_i$
which are interconnected in a chain. It makes clever use of protected
edges in order to delay the separation of root-incompatible connected components
as much as possible, namely, by removing the interconnects between $S_j$ and $S_{j+1}$
in $\N_{i-j}$, i.e., from right $(i)$ to left $(1)$. 

\begin{lemma}
  Given a message adversary $\bD$ and $i$ connected
  subgraphs $\cS_1, \ldots, \cS_i$ of $I(\bD)$ 
  such that for every $1 \le j < i$, 
  the edges of $\bigcup_{j'=1}^j \cS_{j'}$ are protected by
  the communication graphs of $\bigcup_{j'=1}^{j+1} \cS_{j'}$, 
  and $\cS_j$ is connected to $\cS_{j+1}$ in $\N_{i-j}$, 
  it holds that $\cS_1$ is a connected subgraph of~$\N_i$.
  \label{lem:protectedComponents}
\end{lemma}

\begin{proof}
  We show that all edges of $\cS_1$ are in $\N_i$.
  In order to do so, we prove
  by induction on $i' = 1, \ldots, i$, that all edges of
  $\bigcup_{j'=1}^{i-i'+1} \cS_{j'}$ are in $\N_{i'}$.
  
  The base $i'=1$ follows directly from the code of \cref{alg:Nerve}:
  $\N_1 = I(\bD)$, and each graph $\cS_{j'}$ is a subgraph of $I(\bD)$,
  thus every edge of $\bigcup_{j'=1}^{i} \cS_{j'}$ is
  in $\N_1$.
  
  For the inductive step from $i'$ to $i'+1$,
  assume that every edge of $\bigcup_{j'=1}^{i-i'+1} \cS_{j'}$ is present in
  $\N_{i'}$.
  By assumption, 
  every edge $e$ of $\bigcup_{j'=1}^{i-i'} \cS_{j'}$ is protected by
  a communication graph $G$ of $\bigcup_{j'=1}^{i-i'+1} \cS_{j'}$, i.e.,
  by \cref{def:protected}, $\Root(G) \subseteq \ell(e)$.
  As we also assume that $\cS_j$ is connected to $\cS_{j+1}$ in $\N_{i-j}$ for
  $1 \le j < i$, we have that $\cS_{i-i'-j'}$ is connected to $\cS_{i-i'-j'+1}$ in
  $\N_{i'+j'}$ for $0 \le j' < i-i'$.
  Since $\N_{i'+j'}$ is a refinement of $\N_{i'}$,
  $\cS_{i-i'-j'}$ is connected to $\cS_{i-i'-j'+1}$ also in $\N_{i'}$.
  Hence $\bigcup_{j'=1}^{i-i'+1} \cS_{j'}$ is a connected
  subgraph of $\N_{i'}$, and thus $e$ is connected to $G$ in $\N_{i'}$.
  Thus, in $\N_{i'}$, $e$ is in the same connected component as a graph
  $G$ with $\Root(G) \subseteq \ell(e)$ and,
  by Line~\ref{line:guardRM} of \cref{alg:Nerve}, we have $e \in \N_{i'+1}$.
\end{proof}

%
%

We are now ready to prove the main technical result of this section.
For $r=(n-1)\cdot\TD$,
we show how the connectivity 
of two $r$-round communication patterns in $I(\bD^r)$, 
consisting only of communication graphs from certain sets $\bC_1$ and $\bC_2$, respectively,
is related to the connectivity of $\bC_1$ and $\bC_2$ 
in the refined indistinguishability graph $\N_{\TD}$, 
as computed by \cref{alg:Nerve}.

\begin{lemma}
  \label{lem:PatternToGraph}
  Given an oblivious message adversary $\bD$,
  let $\bC$ constitute a connected component of
  $\N_{\TD}$ and let $\bar{\bC} = \bD \setminus \bC$.
  For $r=(n-1)\cdot\TD$, there is no connection in $I(\bD^r)$ between any
  $\sigma_1 \in \bC^r$ and any
  $\sigma_2 \in \bar{\bC} \bD^{r-1}$.
Herein, $\sigma_2 \in \bar{\bC} \bD^{r-1}$ denotes the fact that 
$\sigma_2$ is composed of one graph of $\bar{\bC}$ and then $r-1$ graphs of $\bD$.
\end{lemma}

\begin{proof}
	Assume for a contradiction that there exist $\sigma_1 \in \bC^r$ and 	$\sigma_2 \in \bar{\bC} \bD^{r-1}$ which are connected in $I(\bD^r)$.
	We show that $\bC$ is connected to some node of $\bar{\bC}$ in $\N_{\TD}$, 
	contradicting the fact that $\bC$ is a connected component of $\N_{\TD}$.
	We do so by proving that there are $\TD$ connected subgraphs $\pi_1, \ldots, \pi_{\TD}$ in
	$I(\bD)$, such that each of them intersects $\bC$,
	$\pi_1$ also intersects $\bar{\bC}$,
	and,
	for every $1 \le j < i=TD$, the edges of
	$\bigcup_{j'=1}^j \pi_{j'}$ are protected by the communication graphs of
	$\bigcup_{j'=1}^{j+1} \pi_{j'}$. Moreover, $\pi_j$ is connected to $\pi_{j+1}$ 
in $\N_{i-j}$:	We have that $\pi_j$ and $\pi_{j+1}$ both intersect $\bC$, 
	and since $\bC$ is a connected component in $\N_i$ and $\N_i$ is a refinement of $\N_{i-j}$, all nodes of $\bC$ are in the same connected component of $\N_{i-j}$. We can hence apply 
\cref{lem:protectedComponents}, which reveals that $\pi_1$ is a connected
subgraph of $\N_i$. As $\pi_1$ also intersects both $\bC$ and $\bar{\bC}$, however,
we have the required contradiction.
    
    Let $\tilde{\pi}$ be a path that connects $\sigma_1$ and $\sigma_2$ in $I(\bD^r)$.
    Recall that, for a round $r' \le r$, $\tilde{\pi}(r')$ denotes the round
    $r'$ communication graphs $\sigma(r')$ 
    for all communication patterns $\sigma$ of $\tilde{\pi}$.
    By a repeated application of \cref{cor:removeRound}, we get that $\tilde{\pi}(r')$ is a path 
    that connects $\sigma_1(r') \in \bC$ and $\sigma_2(r') \in \bD$ in $I(\bD)$
    where,
    in particular, $\tilde{\pi}(1)$ connects $\sigma_1(1) \in \bC$ and
    $\sigma_2(1) \in \bar{\bC}$.
    
    We now construct each connected subgraph $\pi_{j}$, $1 \le j \le i$, as a union of paths $\tilde{\pi}(r')$. 
    That is, for some set $R_{j}\subseteq\{1,\ldots,r\}$ of rounds,
    which we will define below,
    we set $\pi_{j} = \bigcup_{r' \in R_{j}} \tilde{\pi}(r')$.
    We denote the largest round of $R_{j}$ as $r^{\ast}_{j} = \max(R_{j})$.

    For $1 \le m < i$, we inductively construct $R_{m+1}$ from
    $R_m$, starting with $R_1 = \{ 1 \}$, i.e., setting $\pi_1 = \tilde{\pi}(1)$.
    We will assert that (1) $r^*_{m+1} \le r^*_{m} + n-1$ and
    (2) the edges of $\pi_{m} = \bigcup_{r' \in R_m} \tilde{\pi}(r')$
    are protected by the communication graphs of
    $\pi_{m+1} = \bigcup_{r' \in R_{m+1}} \tilde{\pi}(r')$.
    For $1 \le m \le \TD$, property (1) together with $r^*_1=1$ guarantees
    $r^*_{m} \le (n-1) (m-1)+1 \le (n-1) \cdot \TD = r$, thus $\tilde{\pi}(r')$ is well-defined for all $r' \in R_m$.
        
    Given $R_m$ for $1 \le m < i$, we construct $R_{m+1}$ as follows:
    By \cref{cor:infProp}, for every edge $e \in \pi_m$, there is a
    round $r_e \le r^*_m + n-1$ such that
    $\tilde{\pi}(r_e)$ contains a graph $G$ with
    $\Root(G) \subseteq \ell(e)$.
    Let $R_{m+1}$ be the set of all such rounds, i.e.,
    $R_{m+1} = \bigcup_{e \in E(\pi_m)} \tilde\pi(r_e)$.
    This ensures (1) by construction and also (2), because every edge $e$ of
    $\pi_{m}$ is protected by a communication graph $G$ of
    $\tilde{\pi}(r_e) \subseteq \pi_{m+1}$.
    Hence, the edges of $\pi_m$ are protected by the communication graphs of
    $\pi_{m+1}$ and so the edges of
    $\bigcup_{k=1}^m \pi_k$ are protected by the communication graphs of   
    $\bigcup_{k=1}^{m+1} \pi_k$.
\end{proof}

We are now ready to state the main theorem of this section, namely, an
upper bound on the decision time complexity of consensus.

\begin{theorem}
  Let $\bD$ be the set of allowed communication graphs of an oblivious message adversary.
  If the connected components of $\N_{\TD}(\bD)$ are root-compatible, then
  consensus is solvable by round $c (n-1) (\TD+1)$, where
  $c$ is the number of connected components in $\N_{\TD}$.
  \label{thm:poss}
\end{theorem}

\begin{proof}
  We show that every connected component of the indistinguishability graph
  $I(\bD^t)$ is broadcastable for $t = c (n-1)(\TD+1)$.
  This implies the theorem, because there exists a mapping for every connected
  component $\cC$ of $I(\bD^t)$ to a process $p$, such that $p$ is a broadcaster in every
  communication pattern of $\cC$. More specifically, as $\cC$ is an indistinguishability 
component, there is, for every process $q$ and every $\sigma \in \bD^t$, a map 
  $\view_\sigma(q, t) \mapsto p$ such that $p$ is a broadcaster in every communication pattern 
of $\sigma$'s connected component in $I(\bD^t)$.
  In every run with a communication pattern from $\cC$, every process has thus
  already learned the input $x_p$ of $p$, which is a valid decision value.
  This decision procedure hence defines a correct consensus algorithm.
  
  It remains to show the broadcastability of the connected components of $I(\bD^t)$.
  Consider a run $\sigma\in\bD^t$, and all the communication patterns $\sigma(i)$, $i=1\ldots,c(n-1)(\TD+1)$ appearing in it.
  By the pigeon-hole principle, at least one connected component $\bC$ of $\N_{\TD}$ must supply $(n-1)(\TD+1)$ of these graphs, when counted with repetitions.
  That is, there is a set $R\subseteq\{1,\ldots,c(n-1)(\TD+1)\}$,
   with
  $|R|=(n-1)(\TD+1)$, such that every $r_i$ with $i \in R$ satisfies
  $\sigma(r_i) \in \bC$. Note that the occurrence of $n-1$ or more
graphs from $\bC$ in $\sigma$ already suffices to
ensure that it is broadcastable by every process $p\in \bigcap_{G \in \bC} \Root(G)$,
i.e., that every process $q\in \Pi$ has $(p,0,x_p) \in \view_{\sigma}(q, t)$.

  Consider another run $\sigma'\in\bD^t$ that is connected to $\sigma$ in $I(\bD^t)$,
  and the communication patterns $\sigma'(i)$ appearing in it.
  If $n-1$ or more of the latter satisfied $\sigma'(r_i) \in \bC$,
  $\sigma'$ would also be broadcastable by $\bigcap_{G \in \bC} \Root(G)$, so assume that this
  is not the case. There are hence at most $n-2$ indices $r_j \in R$ where $\sigma'(r_j) \in \bC$.
  Let $R' \subseteq R$ with $|R'| = (n-1) \cdot \TD$ be the set of indices obtained by
  discarding all these indices $r_j$ from $R$, in addition to discarding some additional 
  indices $\neq 1$ so as to match the desired size of $R'$.
  
  We now construct the $((n-1)\TD)$-round communication patterns
  $\rho, \rho'$ defined by $\rho(j) = \sigma(r_j)$, $\rho'(j) = \sigma'(r_j)$
  for each $j \in R'$. That is, starting out from $\sigma$ and $\sigma'$, which
  are connected in $I(\bD^t)$, we remove all communication rounds not in $R'$.
  By \cref{cor:removeRound}, $\rho$ and $\rho'$ are connected in $I(\bD^{(n-1)\TD})$.
  This, however, contradicts \cref{lem:PatternToGraph}, because
  $\rho \in \bC^{(n-1) \TD}$ and
  $\rho' \in  \bar{\bC}^{(n-1)  \TD} \subseteq \bar{\bC} \times \bD^{(n-1) \TD-1}$ by construction, where $\bC$ is a connected component
  in $\N_{\TD}$ and $\bar{\bC}$ is its complement.
\end{proof}

\section{Lower Bounds}\label{sec:lowerbound}

This section complements our positive results above by studying 
lower bounds. 
In the following, we first establish a relationship between the 
time complexity of the decision procedure and the termination time of
consensus. We will then derive a time complexity
lower bound for the decision procedure, and combine it with the
first result to establish a consensus termination time lower
bound. 

\subsection{Decision complexity and consensus termination time}

First, we present a relationship  (\cref{thm:imposs}) between the number
of iterations of \cref{alg:Nerve} and the time complexity of consensus.
As before, let $\N_i=\N_i(\bD)$ be the refined indistinguishability graph $\N_i$ after $i$ iterations according to \cref{alg:Nerve},
with the set of allowed graphs $\bD$ sometimes omitted for brevity.
Our general strategy 
is to establish that the impossibility of consensus after $i$ rounds is equivalent to the existence of a set of
``broadcast-incompatible'' communication patterns of length $i$, which
are connected to each other in the indistinguishability graph $I(\bD^i)$.
We ensure broadcast-incompatibility by letting this set also contain
communication patterns $G^i$, i.e., 
$i$ repetitions of the same communication graph $G$, taken from a set of
root-incompatible graphs.
Due to the requirement that every decision must be on the input of some broadcaster
whose input value has reached everyone (recall \cref{claim:validity}), this suffices: in $G^i$, the only processes that have reached everyone are the
members of $\Root(G)$, the root component of $G$.
Thus, not all these communication patterns can have led to the same
decision value, which is a contradiction since all connected
round~$i$ communication patterns must have led to the same decision
value if consensus was solved after $i$ rounds.

The core of our proof is contained in \cref{lem:keepConnected}.
It shows that the connectivity of some communication graphs $A, B$ in
$\N_i(\bD)$ implies the connectivity of the communication patterns
$A^i, B^i$ in the indistinguishability graph $I(\bD^i)$. Informally speaking,
it uses an inductive construction for an arbitrary edge $(A, B)$ of
$\N_i$ to show how the corresponding connectivity between $A^i$ and $B^i$
can be preserved for $i$ rounds in $I(\bD^i)$. It crucially relies on the
fact that every $\N_i$ is a refinement of $\N_{i-1}$, with $\N_1$ being
a refinement of $I(\bD)$, which is due to the fact that \cref{alg:Nerve}
iteratively only removes selected edges via
Line~\ref{line:rm} but never adds any edges.

To show that the connectivity of $A^i$ and $B^i$ is preserved, we use the path
in $\N_i$ from $A$ to $\ell(e)$, respectively $B$ to $\ell(e)$, 
to extend the already constructed connected prefixes $A^{i-1}$ and $B^{i-1}$.
Note that this path also occurs in $\N_{i-1}$ due to \cref{cor:nerveEdge}.
To illustrate this, consider a (very simple) example, where we have that $A \sim_p B$ occurs in $\N_{2}$ and
furthermore $p = \Root(C)$ such that $C \sim_{p'} A$ as well as $C \sim_{p''} B$
occur in $\N_{1}$.
In this case, we have the following indistinguishability relation between communication patterns of length~$2$: $A \circ A \sim_{p'} A \circ C \sim_p B \circ C \sim_{p''} B \circ B$. This argument can be applied inductively to establish the indistinguishability relation for communication patterns $A^i$ and $B^i$.

\begin{lemma}
  Let $\cC_i$ be a connected component of $\N_i(\bD)$ and let $A, B$ be
  communication graphs in $\cC_i$.
  Then $A^i$ is connected to $B^i$ in $I(\bD^i)$.
  \label{lem:keepConnected}
\end{lemma}
\begin{proof}
  The lemma holds immediately for $i=1$: As a one-round communication pattern consists of only a single communication graph, $A^1 = A$ and
  $B^1 = B$ are both in the connected component $\cC_1$.

  Thus, we henceforth assume that $i > 1$, and prove the
  following claim by induction on $k$, for $k = 1, \ldots, i$:
  For each edge $(A, B) \in \cC_i$ 
  there is a path $\pi_k$ in $I(\bD^k)$ connecting $A^k$ to $B^k$.
  In addition, for $k < i$, the connected component $\cC_{i-k}$ of
  $A$ and $B$ in $\N_{i-k}$ is such that, for every edge
  $e = (\sigma, \sigma') \in \pi_k$, both the round $k$ communication graphs
  $\sigma(k), \sigma'(k) \in \cC_{i-k}$
  and there is a graph $G_e \in \cC_{i-k}$ such that $\Root(G_e) \subseteq \ell(e)$.
  
  The base, $k=1$, follows because $e = (A, B) \in \cC_i$ implies that
  $(A^1, B^1) \in I(\bD^1)$,
  and by \cref{cor:nerveEdge} there is $G_e \in \cC_{i-1}$ such that
  $\Root(G_e) \subseteq \ell(e)$.
  
  For the step from $k-1$ to $k$, $k>1$, there exists a path $\pi_{k-1} \in I(\bD^{k-1})$ that connects $A^{k-1}$ to $B^{k-1}$. Let
  $e = (\sigma, \sigma') \in \pi_{k-1}$ be an arbitrary edge in $\pi_{k-1}$.
  By the induction hypothesis,  $\sigma(k-1)$, $\sigma'(k-1)\in  \cC_{i-k+1}$
  and there is a graph $G_e \in \cC_{i-k+1}$ with
  $\Root(G_e) \subseteq \ell(e)$. Consequently, there exist paths $\tilde{\pi}_1  = (\Gamma_1 , \Gamma_2, \ldots, \Gamma_m )$ and $\tilde{\pi}_2 = (\Lambda_1 , \Lambda_2, \ldots, \Lambda_{m'})$ in  $\cC_{i-k+1}$ that connect $\sigma(k-1)$ to $G_e$ and $G_e$ to $\sigma'(k-1)$, respectively.
  
  Consider $(\Gamma_j,  \Gamma_{j+1}) \in \tilde{\pi}_1 \subseteq \cC_{i-k+1}$. From \cref{cor:nerveEdge}, we know that $(\Gamma_j,\Gamma_{j+1}) \in I(\bD^1)$, which implies $\sigma \circ \Gamma_j \sim  \sigma \circ \Gamma_{j+1}$. This enables us to prefix $\sigma$ to each communication graph of $\tilde{\pi}_1$, which makes $\sigma \circ \tilde{\pi}_1 = (\sigma \circ \Gamma_1, \sigma \circ \Gamma_2,\ldots, \sigma \circ \Gamma_m)$ a path in $I(\bD^k)$. Following a symmetrical argument, $\sigma' \circ \tilde{\pi}_2 = (\sigma' \circ \Lambda_1, \sigma' \circ \Lambda_2,\ldots, \sigma' \circ \Lambda_{m'})$ is also a path in $I(\bD^k)$. 
  
Moreover, since $\Root(G_e) \subseteq \ell(e)$, it follows from \cref{claim:ignoramus} that $e' = (\sigma \circ G_e, \sigma' \circ G_e) \in I(\bD^k)$. Therefore, $\tilde{\pi}_e = (\sigma \circ \tilde{\pi}_1 , e' , \sigma' \circ \tilde{\pi}_2)$ is a path from $\sigma \circ \sigma(k-1)$ to $\sigma' \circ \sigma'(k-1)$ in $I(\bD^k)$. If we substitute each edge $e \in \pi_{k-1}$ by $\tilde{\pi}_e$, we thus obtain a path $\pi_k$ that connects $A^k$ to $B^k$ in $I(\bD ^k)$.
  
  Now, consider any edge $e' \in \pi_k$. By construction, $e' = (\sigma \circ \Gamma_j, \sigma \circ \Gamma_{j+1})$, or $e' = (\sigma' \circ \Lambda_j, \sigma' \circ \Lambda_{j+1})$ or $e' = (\sigma \circ G_e, \sigma' \circ G_e)$. If $e' = (\sigma \circ \Gamma_j, \sigma \circ \Gamma_{j+1})$, then the round~$k$ communication graphs are $\Gamma_j$ and $\Gamma_{j+1}$. Since $\tilde{\pi}_1 \in \cC_{i-k+1}$, it follows from \cref{cor:nerveEdge} that $(\Gamma_j, \Gamma_{j+1}) \in \cC_{i-k}$, and there exists a communication graph $G_{e'} \in \cC_{i-k}$ with $\Root(G_{e'}) \subseteq \ell((\Gamma_j, \Gamma_{j+1})) = \ell(e')$. A symmetrical argument holds for the case where $e' = (\sigma' \circ \Lambda_j, \sigma' \circ \Lambda_{j+1})$. Finally, if $e' = (\sigma \circ G_e , \sigma' \circ G_e)$, then the round~$k$ communication graphs are both $G_e$, which is in $\cC_{i-k+1}$ by the induction hypothesis. \cref{cor:nerveEdge} guarantees $G_e \in \cC_{i-k}$, and since $\Root(G_e) \subseteq \ell(\sigma,\sigma')$, it follows that $\Root(G_e) \subseteq \ell(\sigma \circ G_e, \sigma' \circ G_e)$. This shows that $G_e$ is a suitable choice for $G_{e'}$, which completes the induction step. 
\end{proof}

\begin{theorem}
  If $\N_i(\bD)$ contains a connected component $\cC_i$ that is not
  root-compatible, then not all processes in all runs of a correct consensus
  algorithm are able to decide after $i$ rounds under the oblivious message
  adversary represented by $\bD$.
  \label{thm:imposs}
\end{theorem}

\begin{proof}

  For the purpose of deriving a contradiction, suppose that the theorem does not
  hold.
  Let $\bS$ be a set of graphs from $\cC_i$ that is not root-compatible.
  By \cref{claim:validity}, for each $G \in \bS$, the decision value in a run
  with communication pattern $G^i$ that consists of $i$ repetitions of $G$
  must be a value $v = x_p$
  for some $p \in \Root(G)$. Since $\bS$ is root incompatible, there exists some $H \in \bS$ such that $x_p$ is not a root value of $H$.
  
  It follows from \cref{lem:keepConnected} that $G^i$ is connected to $H^i$ in $I (\bD ^i)$. Therefore, there is a sequence of runs $(\sigma_1 = G^i, \sigma_2, \ldots, \sigma_m = H^i)$ such that $\sigma_k$ is indistinguishable from $\sigma_{k+1}$. Since all processes decided $v = x_p$ in $G^i = \sigma_1$, by the validity condition of consensus, $\sigma_2$ and inductively all processes in the sequence including $H^i$ should also decide $v= x_p$. Thus, \cref{claim:validity} yields the contradiction that $H^i$ decided a non-broadcasted value.
%
\end{proof}

We conclude by explaining why \cref{thm:imposs} refines the lower bound from~\cite[Theorem 4.10]{CGP15}, which stated
that consensus is impossible if some beta class is not root-compatible, by making the decision time explicit. In fact,
in our terminology, the beta classes are the connected components of
$\N_{\TD}$, where $\TD$ is the smallest round such that
$\N_{\TD} = \N_{\TD-1}$.
Thus, the existence of a root-incompatible beta class is equivalent to
$\N_{\TD}$ containing a root-incompatible connected component.
Note that, since $\N_{\TD} = \N_{\TD-1}$, even if we remove the
termination condition from Line~\ref{line:until} of \cref{alg:Nerve}, for all $\TD' \ge \TD - 1$, we still have that
$\N_{\TD'} = \N_{\TD}$, because, according to \cref{alg:Nerve}, if the set of edges remains the same in an iteration of
$\TD$, then it will remain the same for all future iterations as well.
Thus we can apply \cref{thm:imposs} to show that, in this case, every consensus algorithm has, for every round, a run where some process has not yet decided.
As for an oblivious message adversary with a set of allowed graphs $\bD$, it holds that every infinite communication pattern $\sigma$ with $\sigma|_r \in \bD^r$ for every round $r$ satisfies $\sigma \in \bD^\omega$ (i.e., oblivious message adversaries are limit-closed, see~\cite{WSM19:OPODIS} for details), this implies that there is an infinite run where consensus is not achieved, that is, consensus is indeed impossible.

\subsection{Exponential iteration complexity of the decision procedure}
\label{sec:decProcLB}

As we have seen above, consensus termination time is
related to the iterations of the decision procedure. Informally,
this is due to the fact that the information encoded in the
sequence $\N_1,\dots,\N_{i}$ can be seen as a compact summary 
of the evolution of the indistinguishability relation of the
corresponding communication pattern prefixes.
Thus, a lower bound on the complexity of the decision procedure
immediately gives us a lower bound for the round complexity of 
any consensus algorithm.

In this section, we will show that the decision procedure may take an exponential number of iterations, 
in terms of $n$, until it terminates.
This implies that there are oblivious message adversaries under which consensus is achievable, but reaching it takes exponential time.
As already sketched at the end of \cref{sec:decProc}, we will show this by constructing a specific instance of such a
message adversary, with a set of allowed graphs
$\bD = \{ G_1, \ldots, G_N \}$ of size $N =1.3^n$ (rounded down if necessary), 
whose indistinguishability graph $I(\bD)$ contains the following connected component:
\begin{equation}
G_1 \sim_{R_3} G_2 \sim_{R_4} \ldots \sim_{R_{N+1}} G_N \label{eq:indist}
\end{equation}
Herein, $R_i = \Root(G_i)$ for $1 \le i \le N$, 
and $R_{N+1} \ne \Root(G)$ for all $G \in \bD$. Therefore, $I(\bD)$ contains a path of length $N-1$. Since all edges except the rightmost one are protected,  \cref{alg:Nerve} only removes one edge per iteration, from right to left. More precisely, it holds that $G_1 \sim_{R_3} \ldots \sim_{R_{N-i+1}} G_{N-i} \in \N_i$. Consequently,
$N$ iterations are needed until all edges have disappeared, which establishes our claim.

\noindent \textbf{Informal overview of the definition of $\bD$.}
First, we choose a sequence of sets $\{R_1, \ldots, R_N\}$ that will play the role of root components of $\bD$. We will choose those from the first half $\{p_1, \ldots, p_{n/2}\}$ of the processes only. Each $R_i$ is
chosen to be unique, of the same size $n/12$, and $R_i$, $R_{i+1}$ and $R_{i+2}$ must be be mutually disjoint. Note that we need $N$, i.e., exponentially many such $R_i$.

The first step in the definition of the graph $G_i$ is to make $R_i$ its root component, which is done by fully connecting its members to form a clique and ensuring a path to every other process. However, when doing so, we also need to guarantee that $G_i \sim_{R_{i+2}} G_{i+1}$ are the only indistinguishability relations in $I(\bD)$. We secure this by making sure that every process except for the ones in $R_{i+1}$  and $R_{i+2}$ can distinguish $G_i$ from any other graph $G_j$, $j\neq i$. This is accomplished by adding an outgoing edge from every member of $R_i$ to every process in $\Pi \setminus (R_{i+1} \cup R_{i+2})$, and no other outgoing edge from members of $\{p_1,\dots, p_{n/2}\}$. Since $R_i$ is unique, any process in $\Pi \setminus (R_{i+1} \cup R_{i+2})$ will know if graph $G_i$ is being played: This is immediately obvious for every
process $p$ in the second half $B= =\{p_{n/2+1}, \ldots, p_N\}$, as 
$\In_{G_i}(p) \cap \{p_1, \ldots, p_{n/2}\} = R_i$. For a process $p$
in the ``leftover set'' $L_i=\Pi\setminus(B\cup R_i\cup R_{i+1}\cup R_{i+2}) \subseteq \{p_1, \ldots, p_{n/2}\}$, we have $\In_{G_i}(p) \cap \{p_1, \ldots, p_{n/2}\} = R_i \cup \{p\}$. Since $R_i \cup \{p\}$ is larger than the size of the root components, $p$ knows that it is not part
of the root component, and can hence also uniquely determine $R_i$ and hence the graph $G_i$ being played. \cref{fig: lb-graph} illustrates
this construction.

\begin{wrapfigure}[10]{R}{0.5\textwidth}
	\vspace{-3.5ex}
	\centering
	\includegraphics[scale=0.6]{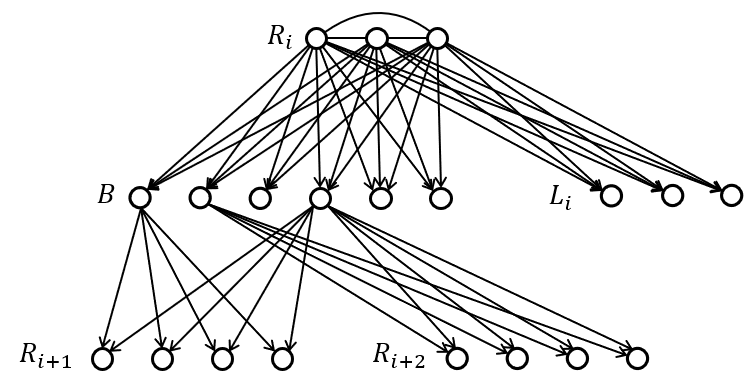}
	\caption{A sketch of the lower bound graph $G_i$}
	\label{fig: lb-graph}
\end{wrapfigure}

However, we must also make sure that all the members of $R_{i+1}$ (resp.\ $R_{i+2}$) consider only $G_i$ and $G_{i-1}$ (resp.\ $G_i$ and $G_{i+1}$) as possibilities for the actually played graph. This means that the in-neighborhood of any process in $R_{i+1}$ (resp. $R_{i+2}$) must be the same in $G_i$ and $G_{i-1}$ (resp.\ $G_i$ and $G_{i+1}$). So far, the processes in $R_{i+1}$ or $R_{i+2}$ do not receive any message from $\{p_1, \ldots, p_n\}$, i.e., the only know that they are either in $R_{i+1}$ or in $R_{i+2}$. To tell them apart, we will connect some processes in $B=\{p_{n/2+1}, \ldots, p_N\}$ to the members of $R_{i+1} \cup R_{i+2}$, in a way that encodes $i+1$ (for the members of $R_{i+1}$) or $i+2$ (for the members of $R_{i+2}$). A process in $R_{i+1} \cup R_{i+2}$ can hence tell from its in-neighborhood whether it belongs to $R_{i+1}$ or $R_{i+2}$. More specifically, abbreviating $B[i]=\{b\in B\mid i_{b-(n/2+1)}=1\}$, where $i_{\ell}$ is the $\ell$\textsuperscript{th} bit in the binary expansion of $i$, we just make sure that $\In_{G_i}(p) = B[i+1]$ for every $p\in R_{i+1}$ and $\In_{G_i}(p) = B[i+2]$ for every $p\in R_{i+2}$. This construction satisfies our indistinguishability requirements: Each process in $R_{i+1}$ (resp.\ $R_{i+2}$) can tell where it belongs to, but do not know whether $G_i$ or $G_{i-1}$ (resp. $G_i$ or $G_{i+1})$ is played.

\noindent \textbf{Formal definition of the root components $R_i$.}
We define the sets $R_i$ by splitting $\{p_1,\ldots, p_{n/2}\}$ into $\{p_1, \ldots, p_{n/4}\}$ and $\{p_{n/4+1},\ldots,p_{n/2}\}$, and construct the sequence $R_1,R_2, \dots$ of root components from partitions of these ranges alternatingly:
Consider all the partitions of $\{p_1, \ldots, p_{n/4}\}$ into three sets of size $n/12$ each. 
Partition number $\ell+1$ constitutes the root components $R_{6\ell+1},R_{6\ell+2},R_{6\ell+3}$.
Similarly, consider consider all the partitions of $\{p_{n/4+1},\ldots,p_{n/2}\}$ into three sets of size $n/12$ each. Set partition $\ell+1$ constitutes the root components $R_{6\ell+4},R_{6\ell+5},R_{6\ell+6}$.

The sequence clearly satisfies, by construction, the following properties:
\begin{enumerate}
	\item
	$|R_i|=n/12$, since we are considering equal-sized partitions of $n/4$ processes into 3 disjoint sets.
	\item
	$R_i\neq R_j$ for $i \neq j$, since all sets of the partitions are unique.
	\item 
	$R_i,R_{i+1},R_{i+2}$ are pairwise disjoint, since they are either members of the same partition and thus disjoint, or one belongs to segment $\{p_1, \ldots, p_{n/4}\}$ and another to segment $\{p_{n/4+1}, \ldots, p_{n/2} \}$.
\end{enumerate}

The length $N$ of the sequence is dominated asymptotically by the number of partitions of $\{p_1, \ldots, p_{n/4}\}$ into three equisized sets,
which is $\frac{1}{6}\binom{n/4}{n/12}\binom{n/6}{n/12}$. 
The definition of the binomial coefficients, along with simple bounds on the factorial function, give
\begin{equation}
\frac{1}{6}\binom{\frac{n}{4}}{\frac{n}{12}}\binom{\frac{n}{6}}{\frac{n}{12}}
= \frac{\left(\frac{n}{4}\right)!}{6\left(\left(\frac{n}{12}\right)!\right)^3}
\geq c\frac{3^{n/4}}{n}>1.3^n \label{eq:asymptotics}
\end{equation}
where $c$ is a constant and $n$ is sufficiently large. It follows that $N$ is exponential with respect to $n$.


\textbf{Formal definition of $G_i$.} We are now ready to define the graphs $G_i$, recall also \cref{fig: lb-graph}. 
Let $B=\{p_{n/2}+1,\ldots,p_n\}$. 
For each $1\leq i\leq N$, the graph $G_i$ is composed of disjoint $5$ node sets: $B,R_i,R_{i+1},R_{i+2}$, where $R_i,R_{i+1},R_{i+2} \subseteq 
\{p_{1},\ldots,p_{n/2}\}$, $B=\{p_{n/2}+1,\ldots,p_n\}$, and
$L_i=\Pi\setminus(B\cup R_i\cup R_{i+1}\cup R_{i+2})$.

Connect every two nodes in $R_i$ by bi-directional edges, forming a clique.
From each node in $R_i$, add a directed edge to each node in $B\cup L_i$.
Finally, for an index $i$, let $B[i]=\{b\in B\mid i_{b-(n/2+1)}=1\}$, where $i_{\ell}$ is the $\ell$\textsuperscript{th} bit in the binary expansion of $i$.
Add an edge from each node of $B[i]$ to each node of $R_{i+1}$, and similarly, from each node of $B[i+1]$ to each node of $R_{i+2}$.


\medskip

We are now ready to show that the so-constructed graphs form an indistinguishability chain according to \cref{eq:indist}.

\begin{claim}
	\label{claim: Bi is good}
	For $1\leq i\leq N$, we have $B[i]\neq \emptyset$,
	and for $1\leq i<j\leq N$, we have $B[i]\neq B[j]$.
\end{claim}

\begin{proof}
	As $N=1.3^n$, we find $\log_2(N)<n/2$, so each 1-bit of $i$ is represented by a process in $B$, which ends up being in $B[i]$. This establishes the second assertion. The first one is now trivial, as $i\geq 1$.
\end{proof}

\begin{claim}
	\label{claim: lb adversary roots}
	For $1\leq i\leq N$, we have $\Root(G_i)=R_i$.
\end{claim}

\begin{proof}
	This is immediate from the graph's definition.
	In $G_i$, all nodes in $R_i$ are connected to one another
  and have no incoming edges from any node not in $R_i$. 
	From each of them, there is a direct edge to all nodes of $B\cup L_i$. Moreover, by \cref{claim: Bi is good}, there is at least one process $b\in B[i]$, so there is a path from each node in $R_i$, through $b$, to each node in $R_{i+1}\cup R_{i+2}$.
\end{proof}

\begin{claim}
	\label{claim: lb adversary indist. chain}
	We have $G_i\sim_{R_{i+2}} G_{i+1}$ for $1\leq i\leq N-1$, 
	and these are the only indistinguishability relations in the graph.
\end{claim}

\begin{proof}
	As we have already explained in the informal overview, in $G_i$, every process that is not in $R_{i+1} \cup R_{i+2}$ can determine that the graph is $G_i$ from its in-neighborhood. This is immediately obvious for processes in $B$, and also possible for a process $p\in L_i$ by observing $|\In_{G_i}(p)|=n/12+1$ and removing itself from it for
determining $R_i$.
	
	For a process $p \in R_{i+1}$ (resp. $R_{i+2}$), it holds by construction that $\In_{G_i}(p) = B[i+1] = \In_{G_{i-1}}(p)$ (resp.\ 
$\In_{G_{i-1}}(p) = B[i+2]  = \In_{G_{i+1}}(p))$, and that  
$G_{i-1}$ (resp. $G_{i+1}$) is the only other graph besides $G_i$ 
where the in-neighborhood of $p$ is the same.
%
%
%
\end{proof}

Our lower bound is now easy to prove.
\begin{theorem}
	There is an oblivious message adversary under which consensus is solvable, but 
	for which the decision procedure takes exponential time to terminate.\label{thm:expiter}
\end{theorem}

\begin{proof}
	Let $\bD=\{G_i\mid 1\leq i\leq N\}$, where  $N=1.3^n$ for $n$ begin sufficiently large for \cref{eq:asymptotics} to hold.
	We consider \cref{alg:Nerve}, and show, by induction on the iteration  number $i$, that after iteration $i$ the graphs $G_1,\ldots,G_{N-i+1}$ constitute the only nontrivial connected component in $\N_i$.
	
	The base case is $\N_1=I(\bD)$, where the graphs $G_1,\ldots,G_{N}$ are connected by \cref{claim: lb adversary indist. chain}.
	For the inductive step $i-1 \to i$, $i>1$, assume $G_1,\ldots,G_{N-i+2}$ is the only nontrivial connected component in $\N_{i-1}$, and consider iteration $i$.
	
	For $G_1,\ldots,G_{N-i+1}$, every two consecutive graphs $G_j,G_{j+1}$ with $1\leq j\leq N-i$ are indistinguishable for a set $R_{j+2}$ by \cref{claim: lb adversary indist. chain}, which is the root component of $G_{j+2}$ by \cref{claim: lb adversary roots}. Since
	 $G_{j+2}$ is in the same connected component as $G_j$ and $G_{j+1}$ in $\N_{i-1}$, the edge $G_j\sim_{R_{j+2}} G_{j+1}$ is incorporated by the algorithm in $\N_i$.
	
	On the other hand, the edge $G_{N-i+1}\sim_{R_{N-i+3}} G_{N-i+2}$ of $\N_{i-1}$ is not added to $\N_i$.
	This is since $R_{N-i+3}$ is the root component of $G_{N-i+3}$, which is not in the nontrivial connected component of $\N_i$. Since all the root components have equal sizes and are distinct, $R_{N-i+3}$ cannot be contained in any other root component either. This completes the induction step.
	
	It follows that the algorithm takes $N=1.3^n$ iterations to complete.
	Upon completion, each connected component of $\N_N$ is a single, root-compatible graph, so consensus is solvable under $\bD$.
\end{proof}





\subsection{Exponential termination time of consensus}\label{sec:ConsensusLB}

From \cref{thm:imposs}, we immediately obtain a termination time lower bound
of $\Omega(\TD)$ for solving consensus. Consequently, the message adversary used in
(the proof of) \cref{thm:expiter}, where $\TD=N=1.3^n$ for sufficiently large $n$,
reveals a lower bound that is exponential in $n$.

We will now adapt the message adversary from \cref{thm:expiter} in \cref{sec:decProcLB} to show that the termination time of consensus may actually be $\Omega(n 1.3^n)$. More
specifically, in the graph $G_i$ shown in \cref{fig: lb-graph}, we replace the direct edges from $R_i$ to $B$ by a path consisting of processes taken from a set $P \subseteq \{p_{n/2+1},\dots,p_n\}$ with $|P| = \Omega(n)$ (i.e., taken away from the original $B$), as illustrated in \cref{fig: lb-graph-extended}. 

\begin{wrapfigure}[10]{R}{0.5\textwidth}
	\vspace{0ex}
	\centering
	\includegraphics[scale=0.6]{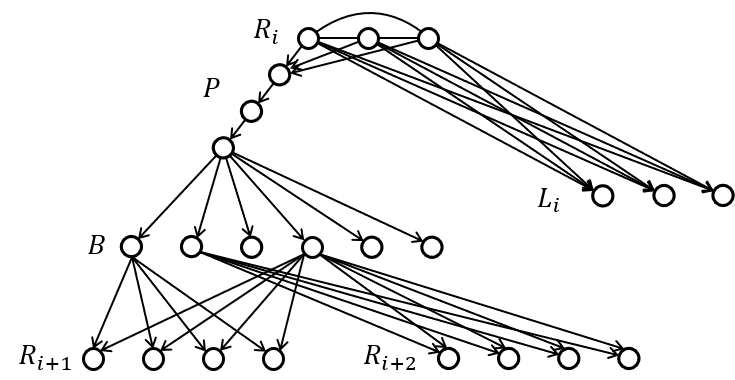}
	\caption{A sketch of the extended lower bound graph $G_i$}
	\label{fig: lb-graph-extended}
\end{wrapfigure}

In more detail, we change the graph construction from \cref{sec:decProcLB}
as follows:
\begin{itemize}
	\item
	$B=\{n/2+1,\ldots,0.9n\}$ and $P=\{0.9n+1,\ldots,n\}$;
	\item
	Add the directed edges $(p,p+1)$ for all $p\in P\setminus\{n\}$;
	\item
	Instead of an edge from each node of $R_i$ to each node of $B$, add an edge from each node of $R_i$ to $h=0.9n+1$, and from $n$ to each node of $B$.
\end{itemize}

Let $h=0.9n$ be the first node on the inserted path. 
Whereas our new construction introduced the additional indistinguishability
$G_i \sim_p G_j$ for all $p \in {(B \cup P) \setminus \{ h \}}$
for
any $G_i, G_j \in \bD$, it does not affect the iteration complexity of the decision procedure, since no $R \subseteq (B \cup P)$ ever occurs
as a root component in a graph of $\bD$.
Thus, all edges $e$ with $\ell(e) \subseteq (B \cup P)$ are removed in the
first iteration of the decision procedure, according to \cref{cor:nerveEdge}.

It is easy to see that \cref{claim: Bi is good} still holds, as we have $\log_2(N)<0.4n$, and \cref{claim: lb adversary roots} holds by construction.
Regarding \cref{claim: lb adversary indist. chain}, the original indistinguishability relations still hold, but are now expanded by additional indistinguishabilities labeled by a process $p\in (P\cup B)\setminus\{h\}$, which are removed in the first iteration of the decision procedure.

The crucial property of our new construction is that any $G_i, G_{i+1}$, when repeated for $0.1n$ rounds, yield indistinguishable communication patterns.

\begin{claim}
  \label{claim:delay}
  $G_{i}^r \sim_{p} G_{i+1}^r \text{ for all } r \le 0.1n \text{ and all } p \in R_{i+2}$.
\end{claim}

\begin{proof}
  Observe that, by construction, we have $\In_{G_i}(p) = \In_{G_{i+1}}(p)$
  for all $p \in X = P \setminus \{ h \} \cup B \cup R_{i+2}$.
  The claim follows, because every path from a process outside $X$ to a process in $R_{i+2}$ has length at least $|P|+1$.
  It thus takes at least $|P|+1$ repetitions of $G_i$, respectively $G_{i+1}$,
  until a process of $\Pi \setminus X$ reached a process of $R_{i+2}$.
  Since $|P| = 0.1n$, in a round $r \le 0.1n$, the nodes of $R_{i+2}$ have
  hence the same view in both $G_{i}^r$ and $G_{i+1}^r$.
\end{proof}

The following \cref{lem:inflation} shows that we can even ``inflate'' arbitrary 
communication patterns of the message adversary from \cref{sec:decProcLB}:

\begin{lemma} \label{lem:inflation}
Consider $(\sigma, \sigma') \in I(\bD^k)$, where $\bD$ is the oblivious message adversary of \cref{sec:decProcLB}. Let $\tilde{\bD}$ be the modified message adversary
of \cref{sec:ConsensusLB}, and $\tilde{\sigma}$ resp.\ $\tilde{\sigma}'$ in
$\tilde{\bD}^{(k 0.1 n)}$ be the communication pattern obtained from replacing every
round $i$ graph $\sigma(i)$ resp. $\sigma'(i)$ according to \cref{fig: lb-graph} 
by $0.1n$ instances of the corresponding graph 
according to \cref{fig: lb-graph-extended}. Then, $(\tilde{\sigma}, \tilde{\sigma}') \in I(\tilde{\bD}^{0.1nk})$.
\end{lemma}

\begin{proof}
We prove, by induction over $k\geq 1$, that (i) the $0.1nk$ prefixes $\tilde{\sigma}|_{0.1nk}$ and $\tilde{\sigma}'|_{0.1nk}$ satisfy $\tilde{\sigma}|_{0.1nk} \sim_R \tilde{\sigma}'|_{0.1nk}$ for $R=\ell(\sigma,\sigma')\neq \emptyset$, and (ii) that 
$\tilde{\sigma}|_{0.1nk} \sim_B \tilde{\sigma}'|_{0.1nk}$ if and only if
$\sigma|_k \sim_B \sigma'|_k$ for the processes $B=\{p_{n/2+1},\dots,p_n\}$. Note carefully that
$\sigma \sim_R \sigma'$ also implies $\sigma|_k \sim_R \sigma'|_k$, as well as
$\sigma(k) \sim_R \sigma'(k)$. As a consequence, there is some $i$ such that, 
for every $k$, either $\sigma(k)=G_i$ and $\sigma'(k)=G_{i+1}$ (or vice versa), 
with $R=R_{i+2}$, or else $\sigma(k)=\sigma'(k)$. 


For the induction basis $k=1$, the only non-trivial case is $\sigma|_1=\sigma(1)=G_i \in \bD$ and 
$\sigma'|_1=\sigma'(1)=G_{i+1} \in \bD$, and $R=R_{i+2}$. From \cref{claim:delay}, we get $\tilde{\sigma}|_{0.1n} \sim_R \tilde{\sigma}'|_{0.1n}$ as needed for (i).
As for (ii), the lenght $0.1n$ of the path $P$ in \cref{fig: lb-graph-extended}
ensures that all processes in $B$ have the same distinguishing power in both the original and in the inflated prefix.
	
For the induction step $k-1 \to k$, $k>1$, we assume for our hypothesis that $\tilde{\sigma}|_{0.1n(k-1)} \sim_R \tilde{\sigma}'|_{0.1n(k-1)}$ and that all processes in
$B$ have the same distinguishing power. Assume for a contradiction for (i) that $\tilde{\sigma}|_{0.1nk} \not\sim_R \tilde{\sigma}'|_{0.1nk}$, i.e., some process $p\in R$ can distinguish the two prefixes. Consider the round $k$ graphs $\sigma(k)$ and $\sigma'(k)$. If $\sigma(k)=\sigma'(k)=G_j \in \bD$, we immediately get
a contradiction, since appending $0.1n$ instances $\hat{G}_j^{0.1n}$ of the corresponding $\hat{G}_j\in\tilde{\bD}$ to both $\tilde{\sigma}|_{0.1n(k-1)}$ and $\tilde{\sigma}'|_{0.1n(k-1)}$ cannot break their indistinguishability for $p$.

So let us assume w.l.o.g.\ $G_i = \sigma(k)$ and $G_{i+1} = \sigma'(k)$ with $R=R_{i+2}$. Since we know from \cref{claim:delay} that the corresponding graphs in $\tilde{\bD}$ ensure $\hat{G}_i^{0.1n} \sim_p \hat{G}_{i+1}^{0.1n}$, the information that allows $p$ to distinguish $\tilde{\sigma}|_{0.1nk}$ and $\tilde{\sigma}'|_{0.1nk}$
was relayed to it from some informed process $q'$ during the last $0.1n$ rounds. Since $R_{i+2}$ only has incoming edges from $B$ in \cref{fig: lb-graph-extended}, there exists an informed process $q \in B$ that relayed this information to $p$ by the last of these rounds. This $q$ must have been informed at the latest in round $0.1nk-1$. Since the path $P$ in \cref{fig: lb-graph-extended} has length $0.1n$, however, $R_i$ (resp.\ $R_{i+1}$) cannot be the source of information that allows $q$ to distinguish $\tilde{\sigma}|_{0.1nk}$ and $\tilde{\sigma}'|_{0.1nk}$. Consequently, $q$ must already have had information to distinguish
$\tilde{\sigma}|_{0.1n(k-1)}$ and $\tilde{\sigma}'|_{0.1n(k-1)}$. From (ii) of
our induction hypothesis, we can infer that this is also true in the original
$\sigma|_{k-1}$ and $\sigma'|_{k-1}$. Since $q$ sends a message 
to $R_{i+2}$ in round $k$ here, this would contradict $\sigma|_k \sim_R \sigma'|_k$,
and therefore completes the induction step for (i).

The induction step for (ii) is trivial, as the processes in $B$ only get information
from the respective root component, either directly (in the original prefix) or delayed via the path $P$ (in the inflated one).
The induction hypothesis hence immediately carries over from $k-1$ to $k$.
\end{proof}

\cref{lem:inflation} immediately gives us the consensus termination time for our new message adversary:

\begin{theorem}
  \label{thm:consensusLB}
  There is an oblivious message adversary for which solving consensus takes $\Omega(n 1.3^n)$ rounds.
\end{theorem}

\begin{proof}
Consider any two indistinguishable communication patterns $\sigma, \sigma'$ of the message adversary of \cref{thm:expiter} on the path between $G_1^{N-1}$ and $G_2^{N-1}$ in $I(\bD)^{N-1}$.
As $\TD=N=1.3^n$, \cref{lem:keepConnected} guarantees that this path exists.
\cref{lem:inflation} immediately provides us with inflated communication patterns $\mu, \mu' \in I(\bD)^{0.1n(N-1)}$ for our new message adversary, which are also
indistinguishable. Together, they form a path between $G_1^{0.1n(N-1)}$ and $G_{2}^{0.1n(N-1)}$ in $I(\bD)^{0.1n(N-1)}$.
Since the root components $R_1=\Root(G_1)$ and $R_{2}=\Root(G_{2})$ are disjoint,
not all processes can have decided by round ${0.1n(N-1)}$, as claimed.
\end{proof}

\section{The Source of Consensus Time Complexity}
\label{sec:source}

In this section, we want to investigate whether the number of iterations $\TD$
of the decision procedure is the sole cause for a large time complexity of
consensus in an oblivious message adversary. 
Before we do so, however, let us briefly reiterate what we have achieved so far.
In \cref{thm:poss} we have seen that consensus can be solved after
$c (n-1)\cdot\TD$ rounds, whereas \cref{thm:consensusLB} revealed that there are
in fact oblivious message adversaries where consensus takes up to $n \TD$ rounds to terminate and $\TD$ may be exponential in $n$.
Thus in these cases a time complexity exponential in $n$ is asymptotically tight
for solving consensus under an oblivious message adversary.
As we know that the consensus time complexity is always at most $c (n-1)\cdot\TD$, and
since we have examples where it is at least $n \TD$,
it might hence be tempting to assume that $\TD$ also determines the termination
time of consensus in all cases.
In this section, we will see that this is not the case, as, to the contrary, there are instances where the decision
procedure terminates after a constant number of iterations while the consensus
time complexity is exponential in $n$.
We now proceed to show how to derive such an instance.

\subsection{A partition of an oblivious message adversary}
\label{subsec:partition}

Before going into the details of how to construct a message adversary with
the desired property of incurring a large time complexity of consensus while
maintaining a low $\TD$, we investigate an abstract property that, if satisfied by an oblivious message adversary $\bD$ for a
parameter $t$, leads to a consensus time complexity in the order of $t$.
Informally, this property is that there exists a partition $\bS_1, \ldots, \bS_t$ of $\bD$ such that $\bS_1$ is connected in the indistinguishability graph $I(\bD)$ and all the edges
that make up this connection are protected by the communication graphs of
$\bS_2$.
Similarly,~$\bS_2$ is connected in $I(\bD)$ and all of the edges in this
connection, along with the ones from $\bS_1$,
are protected by the communication graphs of
$\bS_3$ and so on.
Our claim is that if some $t$ round communication patterns exist that have
no common broadcaster and whose round $1 \le r \le t$ communication graphs are picked from $\bS_r$, then consensus is impossible by round $t$.
The reason for this, as shown in more detail below, is that the set of
communication patterns $\bS_1 \circ \ldots \circ \bS_t$ is connected in the
indistinguishability graph $I(\bD^t)$, because each $\bS_r$ can maintain the
connectivity of $\bS_1 \circ \ldots \circ \bS_{r-1}$ in $I(\bD^{r-1})$ as all
the edges relevant for this connectivity are protected by the communication graphs of~$\bS_r$.

Formally, we express this property as follows:
\begin{definition}
  \label{def:partition}
  Let $\bS_1, \ldots, \bS_t$ be a partition of $\bD$ with
  the following properties, for $1 \le i \le t$:

  \begin{itemize}
    \item[(i)] Each $\bS_i$ is connected. That is, for each $G, G'$ in
    $\bS_i$, there is a path from $G$ to $G'$ in the indistinguishability graph
    $I(\bD)$ that consists only of elements from $\bS_i$.
    \item[(ii)] The edges of the subgraph of $I(\bD)$, induced by
    $\bigcup_{j=1}^{i-1} \bS_j$, are protected by the communication graphs of
    $\bS_i$.
    \item[(iii)] There is no process $p$ such that every communication pattern of
    $\Sigma = \bS_1 \circ \ldots \circ \bS_t$ is broadcastable by $p$.
  \end{itemize}
\end{definition}

Given this partition, we show in \cref{claim:SigmaConnected} below that
$\Sigma$ is connected in $I(\bD^t)$, which
shows that consensus is impossible after $t$ rounds:
If all processes do decide after $t$ rounds in all runs with a communication
pattern of $\Sigma$, they all decide the same value because $\Sigma$ is
connected in $I(\bD^t)$.
Thus, in some run with communication pattern $\sigma\in\Sigma$, the decision is on an
input of a process $p$ even though $\sigma$ is not broadcastable by $p$, which
contradicts \cref{claim:validity}.

\begin{claim}
  The communication patterns of $\Sigma_t = \bS_1 \circ \ldots \circ \bS_t$
  are pairwise connected to each other in $I(\bD^t)$.
  \label{claim:SigmaConnected}
\end{claim}

\begin{proof}
  Let $I(\bD)[\bS]$ denote the subgraph
  of $I(\bD)$, induced by the set of communication graphs~$\bS$.
  We show an even stronger claim, namely that there is a set of edges $E_t$
  that connects $\Sigma_t$ in $I(\bD^t)$ such that for each $e \in E_t$ there is
  an $e' \in I(\bD)[\bigcup_{j \le t} \bS_j]$ with the same label
  $\ell(e) = \ell(e')$.
  We show this by induction on $k$ with
  $\Sigma_k = \bS_1 \circ \ldots \circ \bS_k$.

  The base of the induction $k=1$ follows directly from property (i) of \cref{def:partition}, as $\bS_1$ is connected in $I(\bD)$.
  
  For the step from $k$ to $k+1$, the induction hypothesis is that
  there are edges $E_k$ that connect $\Sigma_k$ such that for every $e \in E_k$
  there is an $e' \in I(\bD)[\bigcup_{j \le k} \bS_j]$ with $\ell(e) = \ell(e')$.
  We use the graphs of $S_{k+1}$ to extend
  $\Sigma_k$ to $\Sigma_{k+1}$ while maintaining the connectivity of $\Sigma_{k+1}$ as follows.
  
  For every $\sigma_1, \sigma_2 \in \Sigma_k$ with
  $e=(\sigma_1, \sigma_2) \in E_k$, we add to $\Sigma_{k+1}$ the extensions
  $\sigma_1 \circ G$ and $\sigma_2 \circ G$ such that $G \in \bS_{k+1}$ and
  $G$ protects $e$.
  Such a communication graph $G$ exists because of property (ii) of
  \cref{def:partition} and
  because there is an edge $e' \in I(\bD)[\bigcup_{j \le k} \bS_j]$ with
  $\ell(e) = \ell(e')$ by hypothesis.
  
  Finally, for all extensions
  $\sigma_1 \circ G, \sigma_2 \circ G$ and
  $\sigma_2 \circ G', \sigma_3 \circ G'$ added to $\Sigma_{k+1}$ in this
  way, by property
  (i) of \cref{def:partition}, there is a path $\pi$ from $G$ to $G'$ in
  $I(\bD)$ that consists only of graphs $G''\in \bS_{k+1}$.
  We can thus add all the communication patterns $\{ \sigma_2 \circ G'' : G'' \in \pi \}$ to $\Sigma_{k+1}$
  as well:
  This maintains the connectivity of $\Sigma_{k+1}$ and ensures the induction
  hypothesis as the path $\pi$ lies entirely in $I(\bD)[\bS_{k+1}]$
  by property (i) of \cref{def:partition}.
\end{proof}

\subsection{An example: choosing the processes}

We now construct a set of communication graphs that can be partitioned in accordance with \cref{def:partition}, into
$t=1.07^n$ sets.
For a set $\Pi$ of $n$ processes, let $m=\lceil\frac{n}{10}\rceil$.
We construct a message adversary with a partition on it, $\bD=\bigcup_{i=1}^{t} \bS_i$,
where each $\bS_i$ is a set of $2i+1$ graphs, denoted
$\bS_i=\{G_{i,j}\mid 1 \le j \le 2i+1\}$.
Each graph $G_{i,j}$ is defined by a partition of the process set $\Pi$ as
\[
	\Pi=
	\begin{cases}
	B\cup R_j \cup U_i \cup U'_i\cup L_{i,j} & \text{ for $j=1$},\\
	B\cup R_j \cup U_i\cup L_{i,j} & \text{ for $j$ even},\\
	B\cup R_j \cup U'_i \cup L_{i,j} & \text{ for $j\geq 3 $ odd}.
	\end{cases}
\]
The process sets are 
	$B=[5m+1, n]$, which is fixed for all $i,j$.
	$R_j$ with $|R_j|=m$, which constitutes the root component of all graphs $G_{i,j}$.
	$U_{i}, U'_i$, with $|U_{i}|=|U'_i|=m$.
	The set $L_{i,j}$ is defined to be the set of all the remaining processes.
\
We choose processes for these sets by induction on $i$, as follows.
For the base, we show how to construct the sets for the communication graphs of
$\bS_1=\{ G_{1,1}, G_{1, 2}, G_{1, 3} \}$.
\begin{itemize}
\item[(b1)] $R_1 = [4m+1, 5m]$
  \item[(b2)] $R_2\subseteq [1, 2m]$, $|R_2|=m$, chosen arbitrarily
  \item[(b3)] $R_3\subseteq [2m+1, 4m]$, $|R_3|=m$, chosen arbitrarily
  \item[(b4)] $U_1\subseteq [2m+1, 4m]$ different from $R_3$
  \item[(b5)] $U'_1\subseteq [1, 2m]$ different from $R_2$
\end{itemize}

We proceed with the inductive step of our construction.
For this we assume that we are given $R_1, \ldots, R_{2i+1}$ and $U_{i}, U'_{i}$, and show how to
construct $R_{2i+2}, R_{2i+3}$ and $U_{i+1}, U'_{i+1}$.

\begin{itemize}
  \item[(s1)] We let $R_{2i+2} = U'_{i}$
  \item[(s2)] We let $R_{2i+3} = U_{i}$
  \item[(s3)] We let $U_{i+1}$ be an arbitrary subset of $[2m+1, 4m]$ of size $m$, different from $R_2, R_4 \ldots R_{2i+2}$
  \item[(s4)] We let $U'_{i+1}$ be an arbitrary subset of
  $[1, 2m]$ of size $m$, different from $R_3, R_5\ldots R_{2i+3}$
\end{itemize}
Note that steps (s1) and (s2) are always possible, as long as the sets $U_i$ and $U'_i$ are defined.
To see that we can repeat step (s3) for~$t$ times, 
note that there are $\binom{2m}{m}$ many ways to choose a set 
$U_{i+1}\subseteq [2m+1,4m]$ of size $m$.
We have 
	\[
\binom{2m}{m} \ge \frac{(2m)^m}{m^m}
= 2^{m} \geq 2^{n/10} > 1.07^n
\]
and the claim follows. 
The claim for (s4) is analogous.

\subsection{An example: the graph structure}
We now show how to combine the sets $R_j, U_i, U'_i, B,$ and $L_{i,j}$ in $G_{i,j}$
to obtain an oblivious message adversary that has a partition as described in
\cref{def:partition} for $t=1.07^n$ (for an illustration, see \cref{fig:Constant-TD}).
While the choice processes of $R_j$ is independent of $i$,
the edges between them in $G_{i,j}$ will depend crucially on $i$.

The graph $G_{i, j}$ always contains a directed cycle
in $R_j$ in increasing order of the process identifiers.
Note that $|R_j|=m$ and each process already has one incoming edge from the preceding process, and thus there are $m-2$ other potential incoming edges we can choose to add.
Hence, there are $m\cdot 2^{m-2}>t$ (for $n$ large enough) possible interconnects for $R_j$, and for each $i$ we choose a different one.

\begin{wrapfigure}[11]{R}{0.47\textwidth}
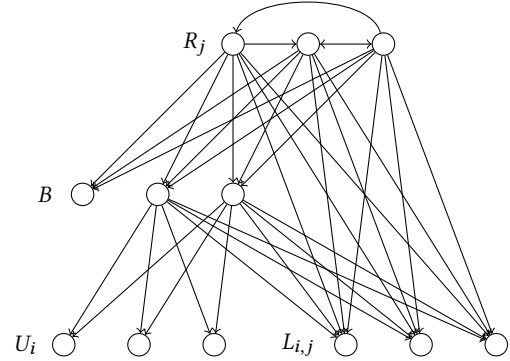

	\vspace{-6ex}
	\ctikzfig{drawings/Constant-TD}
	\vspace{-4ex}
	\caption{Topology of $G_{i,j}$ for even $j$, used to establish an exponential
		consensus time complexity in spite of a constant $\TD$.}
	\label{fig:Constant-TD}
\end{wrapfigure}

We define the other edges of $G_{i,j}$ as follows.
Each graph contains edges 
from all process of $R_j$
to all processes of $B$ and $L_{i,j}$.
For an index $i$, let $B[i]=\{b\in B\mid i_{b-(5m+1)}=1\}$, where $i_{h}$ is the $h$\textsuperscript{th} bit in the binary expansion of $i$.
Note that for $i\neq i'$, $1\leq i,i'\leq t$ we have $B[i]\neq B[i']$, i.e. all the bits of $i$ are represented in $B[i]$, since $\log_2{t}<0.1n$.

The rest of the edges depend on $j$, as follows.
\begin{itemize}
	\item
	For $j=1$, add an edge 
	from each node of $B[i]$ to each node of $U_i \cup U'_i \cup L_{i,j}$.

	\item
	For $j$ even,
	add an edge 
	from each node of $B[i]$ to each node of $U_i  \cup L_{i,j}$.
	
	\item
	For $j\geq 3$ odd,
	add an edge 
	from each node of $B[i]$ to each node of $U'_i  \cup L_{i,j}$.
\end{itemize}

\subsection{An example: properties of the adversary}
Finally, let us establish our main claim, namely that the above construction
indeed yields an oblivious message adversary where the consensus time complexity $t= 1.07^n$
grows exponentially with $n$, yet $\TD = 2$, a constant.
In the remainder of this section, we show these properties for the oblivious message adversary $\bD$ constructed above.
First, we show that $\bD$ partitions as described in \cref{def:partition}.
\begin{claim}
  The sets $\bS_1, \ldots, \bS_t$ are a partition according to
  \cref{def:partition}.
\end{claim}

\begin{proof}
	For property (i), the connectivity of $\bS_i$, pick any
	$G_{i,j} \in \bS_i$.
	We show this graph is indistinguishable to some processes from $G_{i, 1}$, and thus the graph is connected by an edge to $G_{i, 1}$ in $I(\bD)$.
	If $j$ is odd, the in-neighborhood of every process of $U'_i$ is the same in
	$G_{i,j}$ and in $G_{i, 1}$, namely $B[i]$.
	Similarly, if $j$ is even, every process of $U_i$ has $B[i]$ as its in-neighborhood in $G_{i,j}$,
	and this is also the case for $G_{i, 1}$.
	
	To prove property (ii), which states that the communication graphs of $\bS_i$
	protect the edges that were used to connect $\bS_1, \ldots, \bS_{i-1}$,
	it suffices to show that for every $1 \le i' < i$,
	there are communication graphs $G, G' \in \bS_i$ such that
	$\Root(G) \subseteq U_{i'}$ and $\Root(G') \subseteq U'_{i'}$.
	For a given $1 \le i' < i$, note that the graphs $G_{i,2i'+2},G_{i,2i'+3} \in \bS_i$ satisfy
	$R_{2i'+2}=U_{i'}$ and $R_{2i'+3}=U'_{i'}$
	by construction.

	For property (iii), which states that there is no process by which all
	communication patterns of $\Sigma = \bS_1 \circ \cdots \circ \bS_t$ are
	broadcastable, let us investigate the processes that 
	were able to broadcast in
	$(G_{i,2})_{i=1}^t \in \Sigma$ and $(G_{i,3})_{i=1}^t \in \Sigma$.
	We observe that, by (b2) and (b3), for all $1 \le i \le t$,
	$\Root(G_{i,2}) = R_2 \subseteq [1, 2m]$ and
	$\Root(G_{i,3}) = R_3 \subseteq [2m+1, 4m]$ and thus
	$R_2 \cap R_3 = \emptyset$.
	As the broadcasters of $(G_{i,2})_{i=1}^t$ are $R_2$ and
	the broadcasters of $(G_{i,3})_{i=1}^t$ are $R_3$, property (iii) holds.
\end{proof}

\begin{claim}
  The decision procedure terminates after $\TD=2$ iterations on $\bD$.
\end{claim}
\begin{proof}
	First, note that all the roots $R_j$ are contained in $[1,5m]$, while $B=[m5+1,n]$, hence no edge of $I(\bD)$ labeled by only processes of $B$ will be preserved after the first iteration.
	Similarly, we can ignore processes of $B$ in the labels, when considering the preservation of the edges.
	
	We show that in the first iteration of the decision procedure, none of the edges of $I(\bD)$ that connect graphs from different sets in the partition 
	$\bD=\bigcup_{i=1}^{t} \bS_i$
	are preserved.
	Consider $G_{i,j} \in \bS_i$, $G_{i',j'} \in \bS_{i'}$, $i \ne i'$,
	such that $G_{i,j} \sim_\ell G_{i',j'}$.
	Note that in $G_{i,j}$, the processes of $U_i$ (or $U'_i$ if $j$ is odd) and $L_{i,j}$ have $B[i]$ as their incoming edges, while the corresponding processes in $G_{i',j'}$ have $B[i']$, and $B[i]\neq B[i']$,
	so none of $U_i$ (or $U'_i$), $U_{i'}$ (or $U'_{i'}$), $L_{i,j}$ and $L_{i',j'}$ intersect $\ell$.	

	Hence, the only processes in $\ell$ that can occur in a root component of a graph of $\bD$ are processes of $R_j$ and $R_{j'}$.
	Let us study $|R_j\cap R_{j'}\cap \ell|$:
	if $j\neq j'$ then $R_j\neq R_{j'}$ so
	$|R_j\cap R_{j'}\cap \ell|<|R_j|=m$;
	if $j=j'$, then the fact that the choice of interconnects for $R_j$ in $G_{i,j}$ depends on $i$ guarantees that at least one process of $R_j$ is not in $R_j\cap R_{j'}\cap \ell$, and again 
	$|R_j\cap R_{j'}\cap \ell|<m$.
	As any root component $R_{j''}$ of a graph in $\bD$ has $|R_{j''}|=m$, 
	no such root component satisfies $R_{j''}\subseteq R_j\cap R_{j'}\cap \ell$,
	and the edge $\ell$ is not being preserved in the first iteration.
	
	Second, we show that in the second iteration of the decision procedure, none of the edges in $I(\bD)$ that is within a set $\bS_i$ is preserved.
	Assume for contradiction that for some $i$, there are graphs 
	$G_{i,j}, G_{i,j'},  G_{i,j''}\in \bS_{i}$, $j\neq j'$
	such that $G_{i,j}\sim_\ell G_{i,j'}$ and $R_{j''}\subseteq\ell$.
	All the processes of $B, L_{i,j}$ and $L_{i,j'}$ 
	have incoming edges from  $R_j$ (or $R_{j'}$), 
	and since $R_j\neq R_{j'}$ none of these processes appear in~$\ell$.
	Note that $1\leq j''\leq 2i+1$, and 
	the sets $U_i$ and $U'_i$ are chosen to be different from $R_1,\ldots,R_{2i+1}$, 
	which implies $R_{j''}\neq U_i, U'_i$.
		
	If $j=1$, 
	only processes of $U_i\cup U'_i$ can appear in $\ell$.
	This is because in $G_{i,1}$, processes of $R_1$ do not have any incoming edge from $B$, which they have in all other graphs of $\bS_i$, 
	and processes of $L_{i,1}$ have incoming edges from $R_1$, which no process has in any other graph of $\bS_i$.
	Therefore 
	$R_{j''}\subseteq\ell\subseteq U_i\cup U'_i$,
	where
	$U_i\subseteq[2m+1,4m]$ and $U'_i\subseteq[1,2m]$.
	But either $R_{j''}\subseteq [1,2m]$ or 
	$R_{j''}\subseteq [2m+1,4m]$,
	and $|R_{j''}|=|U_i=|U'_i|$,
	so either $R_{j''}=U_i$ or $R_{j''}=U'_i$,
	a contradiction.
	
	If $j$ is even,  
	$R_{j''}\subseteq\ell\subseteq R_j\cup U_i$,
	as any process not in $R_j\cup U_i$ has incoming edges from all processes of $R_j$ in $G_{i,j}$, which it does not have in $G_{i,j'}$.
	We have 
	$R_j\subseteq[1,2m]$ (as $j$ is even) and
	$U_i\subseteq[2m+1,4m]$,
	while either $R_{j''}\subseteq [1,2m]$ or 
	$R_{j''}\subseteq [2m+1,4m]$.
	So, either 
	$R_{j''}=R_j$ or $R_{j''}=U_i$.
	This can only occur if $j=j''$:
	the sets $R_k$ are different for different indices $k$, and $U_i$ is chosen to be different from $R_1,\ldots,R_{2i+1}$.
	The case of $j>1$ odd is analogous,
	and we conclude that $j=j''$ in both cases.
	The same analysis applies for $j'$, 
	and so we have $j=j''=j'$, a contradiction.
\end{proof}

From this, we conclude the main theorem of this section.
\begin{theorem}
  There exists an oblivious message adversary with exponential consensus
  time complexity in spite of a constant iteration complexity $\TD$ of the
  decision procedure.
\end{theorem}

\section{Conclusions}\label{sec:futwork}
This paper presented a simple procedure for deciding whether solving consensus
is possible under a given oblivious message adversary. Whereas it
can be viewed as an early terminating version of the abstract beta class
characterization by Couloma, Godard, and Peters~\cite{CGP15}, our
formulation turned out to be instrumental for characterizing the, to the best
of our knowledge, previously unknown termination time of 
distributed consensus under a given message adversary. We 
discovered a close relation between the number of iterations of the 
decision algorithm and the consensus termination time, and the
importance of the existence and number of root-compatible connected components 
in the refined indistinguishability graph.

Our work opens several interesting avenues for future work.
For example, while we have presented a combinatorial approach, it would be interesting to study the time complexity of the consensus problem from a topological perspective as well. It would further be interesting to fully understand the implications of our approach on distributed information dissemination problems such as broadcast, and explore alternative adversarial models. We also plan to conduct an empirical study of our algorithms to complement the theoretical perspective and analysis presented in this paper. 

\begin{acks}
Research supported by the Austrian Science Fund (FWF) project DELTA (Dependable Network Data Plane for the Cloud), I 5025-N, 
a joint project with Hungarian National Research, Development and Innovation Office NKFIH (co-PI: Gabor Retvari).
\end{acks}

\bibliographystyle{plain}
\bibliography{localbib,lit}



\end{document}

%% file: drawings/Worst-Case-Graphs.tikzstyles
\tikzstyle{myNode}=[fill=white, draw=black, shape=circle]

\tikzstyle{right}=[->]
\tikzstyle{left}=[<-]
\tikzstyle{undir}=[<->]